\documentclass[10pt]{extarticle} 

\usepackage[a4paper,margin=2cm]{geometry} 

\usepackage{newpxtext,newpxmath} 

\usepackage{newpxmath}
\usepackage[T1]{fontenc}
\usepackage[utf8]{inputenc}

\usepackage{breqn} 
\usepackage{empheq} 

\usepackage[colorlinks=true,linkcolor=blue,citecolor=blue,urlcolor=blue]{hyperref}

\usepackage{titlesec}
\titleformat{\section}[block]{\large\bfseries}{\thesection.}{0.5em}{}
\titleformat{\subsection}[block]{\normalsize\bfseries}{\thesubsection.}{0.5em}{}

\usepackage{enumitem}
\setlist{nosep,leftmargin=*}

\usepackage{amsmath, amssymb}
\usepackage{caption}
\usepackage{subcaption}
\usepackage[ruled,linesnumbered,noend]{algorithm2e} 
\usepackage{mathtools}
\usepackage{bbm}
\usepackage{enumitem}
\usepackage{amsthm}
\usepackage{thmtools}
\usepackage{thm-restate}
\usepackage{multirow}
\usepackage{float}
\usepackage[utf8]{inputenc} 
\usepackage[T1]{fontenc}    
\usepackage{hyperref}       
\usepackage{url}            
\usepackage{booktabs}       
\usepackage{amsfonts}       
\usepackage{nicefrac}       
\usepackage{microtype}      
\usepackage{lipsum}
\usepackage{fancyhdr}       
\usepackage{graphicx}       
\graphicspath{{media/}}     

\usepackage{tcolorbox} 

\newtcolorbox[auto counter, number within=section]{examplebox}[1][]{
  colback=blue!5!white, 
  colframe=blue!75!black, 
  fonttitle=\bfseries,
  title=Example~\thetcbcounter, 
  floatplacement=t, 
  #1
}

\newtheorem{theorem}{Theorem}
\newtheorem{lemma}{Lemma}
\newtheorem{definition}{Definition}

\title{Private Synthetic Data Generation in Bounded Memory
}

\author{
  Rayne Holland\thanks{Corresponding author email: rayne.holland@data61.csiro.au.}, Seyit Camtepe, Chandra Thapa, Minhui Xue \\
  CSIRO's Data61 \\
}

\newcommand{\privhp}{\mathtt{PrivHP}}
\newcommand{\pmm}{\mathtt{PMM}}
\newcommand{\superreg}{\mathtt{SRRW}}
\newcommand{\smooth}{\mathtt{Smooth}}
\newcommand{\privtree}{\mathtt{PrivTree}}

\newcommand{\node}{v_{\theta}}
\newcommand{\nodecount}{v.\mathtt{count}}
\newcommand{\counter}{\mathtt{count}}
\newcommand{\subdomain}{\Omega_{\theta}}
\newcommand{\construct}{\mathtt{PrivHP}}
\newcommand{\grow}{\mathtt{GrowPartition}}
\newcommand{\enforce}{\mathtt{EnforceConsistency}}

\newcommand{\laplacedist}{\mathtt{Laplace}}

\newcommand{\diam}{\mathtt{diam}}
\newcommand{\miss}{\mathtt{ConsErr}}
\newcommand{\err}{\mathtt{TotErr}}
\newcommand{\cost}{\mathtt{cost}}
\newcommand{\level}{\mathtt{level}}

\newcommand{\update}{\mathtt{update}}
\newcommand{\query}{\mathtt{query}}
\newcommand{\empiricalx}{\mu_\mathcal{X}}

\newcommand{\treehp}{\mathcal{T}_{\mathtt{PrivHP}}}
\newcommand{\tree}{\mathcal{T}}
\newcommand{\treex}{\mathcal{T}_{\mathcal{X}}}
\newcommand{\treeprune}{\mathcal{T}_{\mathtt{approx}}}
\newcommand{\treetop}{\mathcal{T}_{\mathtt{exact}}}

\newcommand{\tail}{\mathtt{tail}}

\newcommand{\sketch}{\mathtt{sketch}}

\newcommand{\noise}{\Delta_{\mathtt{noise}}}

\newcommand{\approximation}{\Delta_{\mathtt{approx}}}

\newcommand{\lp}{L_{\star}}

\begin{document}
\maketitle

\begin{abstract}
Protecting sensitive information on data streams is a pivotal challenge for modern systems.
Current approaches to providing privacy in data streams can be broadly categorized into two strategies.
The first strategy involves transforming the stream into a private sequence of values, enabling the subsequent use of non-private methods of analysis.
While effective, this approach incurs high memory costs, often proportional to the size of the database.
Alternatively, a compact data structure can be used to provide a private summary of the stream.
However, these data structures are limited to predefined queries, restricting their flexibility.

To overcome these limitations, we propose a lightweight synthetic data generator, $\privhp$, that provides differential privacy guarantees. 
$\privhp$ is based on a novel method for the private hierarchical decomposition of the input domain in bounded memory.
As the decomposition approximates the cumulative distribution function of the input, it serves as a lightweight structure for synthetic data generation.
$\privhp$ is the first method to provide a principled trade-off between accuracy and space for private hierarchical decompositions.
It achieves this by balancing hierarchy depth, noise addition, and   selective pruning of low-frequency subdomains while preserving high-frequency ones, all identified in a privacy-preserving manner.
To ensure memory efficiency, we employ private sketches to estimate subdomain frequencies without accessing the entire dataset. 

Central to our approach is the introduction of a pruning parameter \( k \), which enables an almost smooth interpolation between space usage and utility, and a measure of skew \( \tail_k \), which is a vector of subdomain frequencies containing all but the largest \( k \) coordinates.  
\(\privhp\) processes a dataset \( \mathcal{X} \) using \( M = \mathcal{O}(k \log^2 |\mathcal{X}|) \) space and, on input domain \( \Omega = [0,1]^{d} \), while maintaining \( \varepsilon \)-differential privacy, produces a synthetic data generator that is at distance  
\[
\mathcal{O}\left(\frac{M^{(1-\frac{1}{d})}}{\varepsilon n} + \frac{||\tail_k(\mathcal{X})||_1}{M^{1/d}n} \right) 
\]
from the empirical distribution in the expected Wasserstein metric.  
Compared to the state-of-the-art, $\pmm$, which achieves accuracy \( \mathcal{O}((\varepsilon n)^{-1/d}) \) with memory \( \mathcal{O}(\varepsilon n) \), our method introduces an additional approximation error term of  \( \mathcal{O}(||\tail_k(\mathcal{X})||_1/(M^{1/d}n)) \), but operates in significantly reduced space.  
Additionally, we provide interpretable utility bounds that account for all error sources, including those introduced by the fixed hierarchy depth, privacy noise, hierarchy pruning, and frequency approximations.
\end{abstract}

\section{Introduction}
\label{sec:introduction}

A data stream solution addresses the problem of analyzing large volumes of data with limited resources.
Numerous techniques have been developed to enable real-time analytics  with minimal memory usage and high throughput~\cite{cormode2013unifying, jowhari2011tight, monemizadeh2010}. 
However, if the stream contains sensitive information, privacy concerns become paramount \cite{epasto2023differentially, wang2021continuous}. 
In such cases, a data stream solution must optimize resource efficiency while balancing utility and strong privacy protection.

The concept of \textit{differential privacy} has emerged as the prevailing standard for ensuring privacy in the analysis of data streams. 
It guarantees that an observer analyzing the outputs of a differentially private algorithm is fundamentally limited, in an information-theoretic sense, in their ability to infer the presence or absence of any individual data point in the stream.
Methods for supporting differentially private queries on data streams can be broadly categorized into two strategies. 
The first strategy transforms the data stream into a differentially private sequence of values or statistics, which can subsequently be processed and queried by non-private data structures \cite{chen2017pegasus, kumar2024privstream, perrier2018private, wang2021continuous}.  
However, these approaches can incur high memory costs, limiting their applicability in many streaming contexts.
The second approach involves constructing specialized data structures, with small memory allocations, that provide differentially private answers to specific, predefined queries  \cite{biswas2024differentially, lebeda2023better, zhao2022differentially}. 
Although this approach is memory efficient, it restricts the range of queries to those selected in advance, reducing flexibility.

To address these limitations, we propose a lightweight synthetic data generator, $\privhp$ (Private Hot Partition), that provides differential privacy and operates with high throughput and in bounded memory. 
Our generator produces private synthetic data that approximates the distribution of the original data stream.
This synthetic data can be used for any downstream task without additional privacy costs.
Therefore, $\privhp$ both protects sensitive information in bounded memory \textit{and} supports a broad range of queries in resource-constrained environments.

\subsection{Problem and Solution}

Mathematically, the problem of generating synthetic data can be defined as follows.
Let $(\Omega, \rho)$ be a metric space and consider a stream $\mathcal{X}=(X_1, \ldots, X_n) \in \Omega^n$.
Our goal is to construct a space and time-efficient randomized algorithm that outputs synthetic data $\mathcal{Y}=(Y_1,\ldots, Y_m) \in \Omega^m$, such that the two empirical measures
\begin{align*}
    \mu_{\mathcal{X}} = \frac{1}{n}{\sum_{i=1}^n} \delta_{X_i} \quad \text{and} \quad \mu_{\mathcal{Y}} = \frac{1}{m}{\sum_{i=1}^m} \delta_{Y_i}
\end{align*}
are close together.
Moreover, the output $\mathcal{Y}$ should satisfy differential privacy.

Our approach adapts recent advancements in private synthetic data generation (in a non-streaming setting) that utilizes a \textit{hierarchical decomposition} to partition the sample space \cite{he2023algorithmically, he2023differentially}. 
At a high level, these methods work in the following way:
\begin{enumerate}
    \item Hierarchically split the sample space into smaller subdomains;
    \item For each subdomain, count how often items from the dataset appear within it;
    \item To ensure privacy, add a carefully chosen amount of random noise to each frequency count;
    \item Use these noisy counts to construct a probability distribution so that items from each subdomain can be sampled with a probability proportional to their noisy frequency. 
\end{enumerate}
The accuracy of this approach depends on the granularity of the partition.
Creating more subdomains of smaller diameter can improve the approximation of the real distribution, but requires more memory. 
Therefore, the challenge of constructing a high-fidelity partition on a stream is this balance between granularity and memory.

To meet this challenge, $\privhp$ adopts a new method for hierarchical decomposition that prunes less significant subdomains while prioritizing subdomains with a high frequency of items. 
In addition, we use sketching techniques to approximate frequencies at deeper levels in the hierarchy, allowing for effective pruning without requiring access to the entire dataset.
Similar to He \textit{et al.} \cite{he2023algorithmically}, to maintain privacy, we perturb the frequency counts of nodes in the (pruned) hierarchy.

\subsection{Main Result}

We measure the utility of a synthetic data generator $\tree$ by $\mathbb{E}[W_1(\mu_{\mathcal{X}}, \tree)]$, where $W_1(\mu_{\mathcal{X}}, \tree)$ is the 1-Wasserstein metric, and $\mathbb{E}$ is taken over the randomness of the algorithm generating $\tree$.
Our results explore trade-offs between utility and performance.
We are interested in quantifying the cost, in utility, of supporting synthetic data generation under resource constraints. 
To quantify the cost of pruning we introduce the vector $\tail_k^r$, which is the vector of subdomain frequencies, at level $r$ in the hierarchy, with the highest $k$ coordinates set to 0.
The norm $||\tail_k^r||_1$ is small for skewed inputs and can even be zero on sparse inputs.
Our utility bound, expressed as a function of the memory allocation,
is formalized in the following result on the hypercube~$\Omega = [0,1]^d$.

\begin{theorem}
    When $\Omega = [0,1]^d$, for pruning parameter $k$, $\privhp$ can process a stream $\mathcal{X}$ of size $n$ in $M=\mathcal{O}(k\log^2(n))$ memory and $\mathcal{O}(\log (\varepsilon n))$ update time.
    $\privhp$ can subsequently output a $\varepsilon$-differentially private synthetic data generator $\treehp$, in $\mathcal{O}(M \log n)$ time, such that
    \begin{align*}
        \mathbb{E}[W_1(\mu_{\mathcal{X}}, \treehp )] = \begin{cases}
             \mathcal{O}\left( \frac{\log^2(M)}{\varepsilon n} + M^{-1}\cdot \frac{||\tail_k^{\varepsilon n}(\mathcal{X})||_1}{n} \right) & if d = 1  
            \\
            \mathcal{O}\left(\frac{M^{(1-\frac{1}{d})}}{\varepsilon n}+{M^{-\frac{1}{d}}}\cdot \frac{||\tail_k^{\varepsilon n}(\mathcal{X})||_1}{n}\right) & if d\geq 2 
    \end{cases},
    \end{align*}
    \label{thm:introduction}
\end{theorem}
\noindent
Our work makes several key contributions. 
First, we present a new method for private hierarchical decomposition that provides the first known trade-off between accuracy and space. 
Our approach achieves this by carefully balancing hierarchical depth, noise addition, and pruning, allowing for efficient representations without sacrificing utility. 
The key innovation is the introduction of the pruning parameter $k$,
which controls the accuracy-space trade-off and provides a level of flexibility unavailable in prior work.

A central contribution of our work is the accuracy analysis. 
The key challenge lies in understanding how error (from both privacy noise and frequency approximation) propagates down the hierarchy and affects pruning decisions.
We formalize how this error is propagated and provide interpretable utility bounds (Theorem~\ref{thm:utilityHP}) that account for the different sources of error.
These bounds apply to \textit{any} metric space as an input domain. 
Thus, our techniques are applicable to a broad range of domains, such as geographic coordinates or the IPv4 address space.

Leveraging the constructed hierarchy, we introduce the first private synthetic data generator with provable space-utility trade-offs for any metric space. 
This enables the generation of high-fidelity, privacy-preserving synthetic datasets, in bounded memory, that are suitable for various downstream data analysis tasks and applications. 
A comprehensive comparison with existing methods for private synthetic data generation is provided in Table~\ref{tab:performance_results}.

\subsection{Organization of the Paper}
As background, Section~\ref{sec:related_work} covers related work and Section~\ref{sec:preliminaries} establishes the relevant preliminaries.
Section~\ref{sec:compact_hd} introduces our method for compact hierarchical decompositions on the stream.
Section~\ref{sec:private_synthetic} covers synthetic data generation using hierarchical partitions and supplies the main results.
Section~\ref{sec:m_utility} introduces our method for measuring utility. 
Lastly, Section~\ref{sec:proof} contains the proof of our utility bound for general input domains.

\section{Related Work}
\label{sec:related_work}

\begin{table}[]
    \centering
    \begin{tabular}{|c||c|c|c|} \hline
        \multirow{2}{*}{\textbf{Method}}
                & \multicolumn{2}{c|}{\textbf{Accuracy}} 
                    & \multirow{2}{*}{\textbf{Memory}} \\ \cline{2-3}
                & $\Omega = [0,1]$ 
                    & $\Omega = [0,1]^d, d\geq 2$ & \\ \hline \hline
        \rule{0pt}{1.6em} $\smooth$ \cite{wang2016differentially}
                & \multicolumn{2}{c|}{$\mathcal{O}\left(\varepsilon^{-1}n^{-K/(2d+K)}\right)$}
                    &  $\mathcal{O}(dn)$ \\ \hline
        \rule{0pt}{1.6em} $\superreg$ \cite{boedihardjo2024private}
                & \multicolumn{2}{c|}{$\mathcal{O}\left(\left({\log^{\frac{3}{2}}(\varepsilon n)}{(\varepsilon n)^{-1}}\right)^{1/d}\right)$}
                    & $\mathcal{O}(dn)$ \\ \hline
        \rule{0pt}{1.6em} $\pmm$ \cite{he2023algorithmically} 
                & $\mathcal{O}\left({\log^2(\varepsilon n)}{(\varepsilon n)^{-1}}\right)$ 
                    & $\mathcal{O}\left((\varepsilon n)^{-1/d}\right)$ 
                        & $\mathcal{O}(\varepsilon n)$ \\ \hline
        \rule{0pt}{1.6em} $\privhp$ 
                & $\mathcal{O}\left( \frac{\log^2(M)}{\varepsilon n} + M^{-1}\cdot \frac{||\tail_k^{\varepsilon n}||_1}{n} \right)$ 
                    & $\mathcal{O}\left(\frac{M^{(1-\frac{1}{d})}}{\varepsilon n}+{M^{-\frac{1}{d}}}\cdot \frac{||\tail_k^{\varepsilon n}||_1}{n}\right)$ 
                        & $M=\mathcal{O}( k\log^2  n)$ \\ \hline
    \end{tabular}
    \caption{Performance of $\privhp$ vs prior work.
    Results are presented for input domains $\Omega=[0,1]$ and $\Omega=[0,1]^d$ ($d\geq 2$).
    Accuracy is measured by the expected 1-Wasserstein distance.
    The utility guarantee for $\smooth$ is restricted to smooth queries with bounded partial derivatives of order $K$.  
    }
    \label{tab:performance_results}
\end{table}

\subsection{Private Hierarchical Decomposition}

Many methods for hierarchical decomposition have been adapted to the context of differential privacy \cite{cormode2012differentially, kumar2024privstream, qardaji2013understanding,yan2020differential,zhang2016privtree}.
Static solutions, such as $\mathtt{PrivTree}$ \cite{zhang2016privtree}, require full access to the dataset and are not suitable for streaming.
The dynamic decomposition introduced by $\mathtt{PrivStream}$ \cite{kumar2024privstream} adapts $\mathtt{PrivTree}$ to the stream.
However, $\mathtt{PrivStream}$ is supported by a fixed and potentially large hierarchical partition of the input domain that must be stored locally with exact counts~\cite{kumar2024privstream}.
Therefore, it incurs large memory costs and is not suited to resource-constrained settings.
We overcome this limitation by adopting pruning and sketching techniques to summarize deeper levels of the hierarchy.
In addition, $\mathtt{PrivStream}$ does not provide any utility guarantees.

Biswas \textit{et al.} introduced a streaming solution for hierarchical heavy hitters \cite{biswas2024differentially} that can support a private hierarchical decomposition on the stream.
One key difference in our approach is the choice of private sketch. 
The hashing-based private sketch \cite{pagh2022improved} employed by $\privhp$  has a better error guarantee than the counter-based sketch \cite{lebeda2023better} used by Biswas \textit{et al.} \cite{biswas2024differentially}.
Further, as the error of the hash-based sketch can be expressed in terms of the tail of the dataset it composes nicely with  hierarchy pruning.

\subsection{Privacy on Data Streams}

A common strategy for supporting privacy on streams is to transform the stream into a differentially private sequence of values \cite{chen2017pegasus, perrier2018private, wang2021continuous}.
While this enables query flexibility, current methods provide no bounds on the memory allocation, which can be proportional to the size of the stream.
This limits their application in resource-constrained environments.
In the traditional data steam model, where memory is sublinear in the size of the database, the current approach for protecting streams is to construct specialized data structures, with small memory allocations, that provide differentially private answers to specific, predefined queries  \cite{biswas2024differentially, lebeda2023better, zhao2022differentially}.
However, these methods lack query flexibility.

Alabi \textit{et al.} provide a method for private quantile estimation in bounded memory \cite{alabi2022bounded}.
A quantile estimator can be used to generate synthetic data approximating the distribution of the input dataset.
This is achieved by sampling a value uniformly in [0,1] and returning the quantile.
However, their method only works for finite and ordered input domains and, thus, does not extend to general metric spaces.

\subsection{Non-Streaming Private Synthetic Data Generation}

The problem of generating private synthetic data from \textit{static} datasets, especially in relation to differential privacy, has been explored in depth. 
The challenge of this problem was demonstrated by Ullman and Vadhan, who showed that, given assumptions on one-way functions, generating private synthetic data for all two-dimensional marginals is NP-hard on the Boolean cube \cite{ullman2011pcps}. 
A subsequent portion of research has since focused on guaranteeing privacy for specific query sets \cite{barak2007privacy, boedihardjo2023covariance, dwork2015efficient, liu2021iterative, thaler2012faster, vietri2022private}. 

The utility for private synthetic data is measured by the expected 1-Wasserstein distance.
Wang \textit{et al.}~\cite{wang2016differentially} addressed private synthetic data generation on the hypercube $[0, 1]^d$. They introduced a method ($\smooth$) for generating synthetic data that comes with a utility guarantee for smooth queries with bounded partial derivatives of order $K$, achieving accuracy $\mathcal{O}(\varepsilon^{-1}n^{-K/(2d+K)})$.
More recently, Boedihardjo \textit{et al.} \cite{boedihardjo2024private} proved an accuracy lower bound of $\mathcal{O}(n^{-1/d})$.
They also introduced an approach ($\superreg$) based on super-regular random walks with near-optimal utility of $\mathcal{O}(\log^{3/2}(\varepsilon n)(\varepsilon n)^{-1/d})$. 
Subsequently, He \textit{et al.}~\cite{he2023algorithmically} proposed an approach $\pmm$, based on hierarchical decomposition, that achieves optimal accuracy (up to constant factors) for $d\geq 2$.
These methods provide a combination of privacy and provable utility.
However, they do not consider resource constraints.
In contrast, our approach provides meaningful trade-offs between utility and resources.
A summary of accuracy vs. memory for prior work is presented in Table~\ref{tab:performance_results}.

\section{Preliminaries}
\label{sec:preliminaries}
\subsection{Privacy}

Differential privacy ensures that the inclusion or exclusion of any individual in a dataset has a minimal and bounded impact on the output of a mechanism.
Two streams $\mathcal{X}$ and $\mathcal{X}^{\prime}$ are \textit{neighboring}, 
denoted $\mathcal{X} \sim \mathcal{X}^{\prime}$,
if they differ in one element.
Formally, $\mathcal{X} \sim \mathcal{X}^{\prime}$ if there exists a unique $i$ such that
$x_i \neq x_i^{\prime}$.
The following definition of differential privacy is adapted from Dwork and Roth~\cite{dwork2014algorithmic}.
\begin{definition}[{Differential Privacy -- 1-Pass}] A randomized mechanism $\mathcal{M}$ satisfies $\varepsilon$-differential privacy in a 1\emph{-pass} setting if and only if, for all pairs of neighboring streams $\mathcal{X} \sim \mathcal{X}^{\prime} \in \Omega^*$ and all measurable sets of outputs $Z \subseteq \mathtt{support}(\mathcal{M})$, it holds that 
\[ 
\textup{\textsf{Pr}}[\mathcal{M}(\mathcal{X}) \in Z] \leq e^{\varepsilon} \, \textup{\textsf{Pr}}[\mathcal{M}(\mathcal{X}^{\prime}) \in Z].
\]
\end{definition}
\noindent
This states that the output after the stream is processed is differentially private.
This is in contrast to continual observation, where the output is published after each stream update.
Our focus is on the 1-pass model, but our method can be adapted to continual observation by replacing the counters and sketches with their continual observation counterparts.

The Laplace mechanism is a fundamental technique for ensuring differential privacy by adding noise calibrated to a function's sensitivity.
Let $\triangle_p(f) = \max_{\mathcal{X} \sim \mathcal{X}^{\prime}} \lVert f(\mathcal{X}) - f(\mathcal{X}^{\prime}) \rVert_p$ denote the  $p$-sensitivity of the function $f$.
\begin{lemma}[Laplace Mechanism]
Let \( f \) be a function with \( L_1 \)-sensitivity \( \triangle_1(f) \). The mechanism
\[
\mathcal{M}(\mathcal{X}) = f(\mathcal{X}) + \laplacedist\left(\frac{\triangle_1(f)}{\varepsilon}\right),
\]
satisfies \(\varepsilon\)-differential privacy, where \(\laplacedist\) is a Laplace distribution with mean 0 and scale parameter \(\frac{\triangle_1(f)}{\varepsilon}\).
\label{lem:laplace_mechanism}
\end{lemma}

A significant property of differential privacy is that it is invariant under post-processing. 
This means that applying any deterministic or randomized function to the output of an $\varepsilon$-differentially private mechanism does not decrease its privacy guarantee. 

\begin{lemma}[Post-Processing]
If $\mathcal{M}$ is an $\varepsilon$-differentially private mechanism and $g$ is any randomized mapping, then $g \circ \mathcal{M}$ is also $\varepsilon$-differentially private.
\label{lem:post_processing}
\end{lemma}

Differential privacy is also preserved under composition. 
When multiple differentially private mechanisms are applied to the same data, the total privacy loss accumulates.
The composition property quantifies this cumulative privacy loss, providing bounds for combining mechanisms.

\begin{lemma}[Basic Composition]
If $\mathcal{M}_1$ and $\mathcal{M}_2$ are $\varepsilon_1$- and $\varepsilon_2$-differentially private mechanisms, respectively, then the mechanism defined by their joint application, $(\mathcal{M}_1, \mathcal{M}_2)$, is $(\varepsilon_1 + \varepsilon_2)$-differentially private.
\label{lem:composition}
\end{lemma}

\subsection{Utility}

The utility of the output is measured by the expectation of the 1-Wasserstein distance between two measures: 
\begin{align}
   \mathcal{W}_1(\mu_{\mathcal{X}}, \mu_{\mathcal{Y}}) = \sup_{\textsf{Lip}(f)\leq 1} \left( \int f d \mu_{\mathcal{X}} - \int f d \mu_{\mathcal{Y}}   \right),
   \label{eqn:wasserstein}
\end{align}
where the supremum is taken over all 1-Lipschitz functions on $\Omega$.
Since many machine learning algorithms are Lipschitz \cite{meunier2022dynamical,von2004distance}, Equation~\ref{eqn:wasserstein} provides a uniform accuracy guarantee for a wide range of machine learning tasks performed on synthetic datasets whose empirical measure is close to $\mu_{\mathcal{X}}$ in the 1-Wasserstein distance.

\subsection{Sketching}

Our generator relies on a partition of the sample space. 
Prior works, such as $\privtree$ and $\pmm$, construct partitions using \textit{exact} frequency counts, which require access to the full dataset. 
Under the constraint of a sublinear memory allocation, we instead use approximate frequency counts that are {private}. 
We adopt private hash-based sketches~\cite{pagh2022improved,zhao2022differentially} to meet these needs.

A sketch performs a random linear transformation $A$ to embed a vector $v \in \mathbb{R}^{n}$ into a smaller domain $Av \in \mathbb{R}^{j\times w}$.
When the embedding dimension $j\times w$ is small, memory is reduced at the cost of increased error. 
A \textit{private} sketch adds noise to the transformation to make the distribution of the sketch indistinguishable on neighboring inputs. 
Common examples are the Private Count Sketch and Private Count-Min Sketch \cite{pagh2022improved,zhao2022differentially}. 
They share a similar structure but differ in their update and query procedures.
We first overview non-private variants of sketches and then demonstrate how to apply differential privacy.

The Count-Min Sketch \cite{cormode2005improved} is a $j \times w$ matrix of counters, $C(v) \in \mathbb{R}^{j\times w}$, determined by random hash functions $h_1, \ldots, h_j:[n] \rightarrow [w]$, where each $h_i$ maps entries $x \in v$ into a bucket $h_i(x)$ within row $i$.
For $(i, k) \in [j] \times [w]$, each bucket is defined as:
\[
    C(X)_{i,k} = \sum_{x \in v} v_x \cdot \mathbbm{1}(h_i(x)=k),
\]
where $\mathbbm{1}(\xi)$ indicates event $\xi$ and $v_x$ is the count at entry $x$ in $v$. 
Thus, each entry $x$ is added to buckets $(i, h_i(x))$ for all $i \in [j]$, creating $j$ hash tables of size $w$.
The update procedure is visualized in Figure~\ref{fig:cms}.
The estimator $\hat{v}_x = \min \{ C[i][h_i(x)] \mid i \in [j]\}$ combines row estimates by taking the minimum value, filtering out collisions with high-frequency items.
Our results use the following bound on the expected error in a Count-min Sketch, where the influence of high frequency items decays exponentially with the number of rows.
\begin{restatable}{lemma}{lemsketch}
    For input vector $v$, the estimation error of a Count-min Sketch, with width $2w$ and depth $j $, satisfies, $\forall x \in v$, 
    \[
        \mathbb{E}[\hat{v}_x - v_x ] \leq \frac{||\tail_w(v)||_1 + 2^{-j+1} ||v||_1}{w} 
    \]
    where $\tail_w(v)$ is the vector $v$ with the $w$ largest coordinates removed. 
    \label{lem:cms_expected_error}
\end{restatable}
The proof can be found in Appendix~\ref{sec:app_sketch}.
Similar to prior work \cite{pagh2022improved}, our result relies on the use of fully random hash functions.
Significantly, our privacy guarantee does not depend on this assumption.

\begin{figure}
    \centering
    \includegraphics[width=0.5\linewidth]{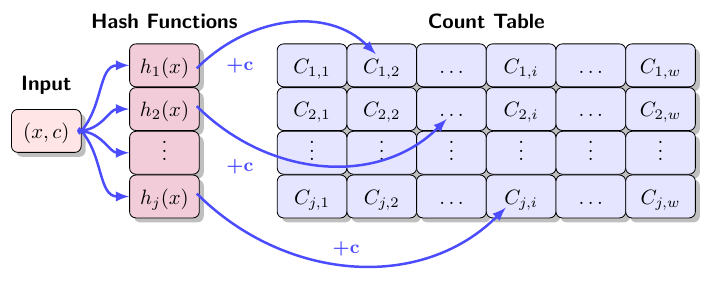}
    \caption{
    The \textsf{update} procedure for a Count-Min Sketch. 
    It follows that~$h_j(x)=i$.
    }
    \label{fig:cms}
\end{figure}

\subsection{Private Release of Sketches}
For private release, sketches must have similar distributions on neighboring vectors. 
In \textit{oblivious} approaches, we sample a random vector $v \in \mathbb{R}^{j \times w}$ independent of the data and release $C(X)+v$. 
The sampling distribution depends on the {sensitivity} of the sketch.
Since sketches are linear, for neighboring inputs $X \sim X^{\prime}$, we have $ C(X) - C(X^{\prime}) = C(X - X^{\prime})$.
As neighboring inputs have sensitivity $\triangle_1(X - X^{\prime}) =1$,
a sketch $C$ has sensitivity proportional to its number of rows. 
Therefore, $\triangle_1(C) =j$ and $C(X)+v$ achieves $\varepsilon$-differential privacy for $v\sim \laplacedist^{j\times w}(j\varepsilon^{-1})$ by Lemma~\ref{lem:laplace_mechanism}.

\section{Private Hierarchical Decomposition in Bounded Memory}
\label{sec:compact_hd}
A hierarchical decomposition recursively splits a sample space into smaller subdomains.
Each point of splitting refers to a \textit{level} in the hierarchy.
Formally, for a binary partition of the sample space $\Omega$, the first level of the hierarchy contains disjoint subsets $\Omega_0, \Omega_1 \subset \Omega$,
such that $\Omega_0 \cup \Omega_1 = \Omega$.
Accordingly, a hierarchical decomposition $\tree$ of depth $L$ is a family of subsets $\Omega_{\theta}$ indexed by $\theta \in \tree \subseteq \{0,1\}^{\leq L}$, where
\[
    \{0,1\}^{\leq L} := \{0,1\}^0 \cup \{0,1\}^1 \cup \cdots \cup \{0,1\}^L.
\]
By convention the cube $\{0,1\}^0 = \varnothing$.
For $\Omega_{\theta}$ with $\theta \in \{0,1\}^l$, we call $l$ the level of $\Omega_{\theta}$.
The leaves of the decomposition $\tree$ form a partition of the sample space. 
To generate synthetic data, any decomposition $\tree$ can be used to form a sampling distribution.
A synthetic point can be constructed by (1) selecting a leaf subdomain $\subdomain \subseteq \Omega$ with probability proportional to its cardinality and (2) conditioned on this selection, choosing a point $y\in \subdomain$ uniformly at random.

Our lightweight generator is based on a \textit{pruned} hierarchical decomposition.
The quality of the generator depends on the granularity of the subsets in the partition.
That is, allowing more subsets of smaller area leads to a sampling distribution closer to the empirical distribution of the input.
Thus, due to the memory cost of storing a more fine-grained partition, we observe a trade-off between utility and space.
To balance this trade-off, we aim to construct a hierarchical decomposition that provides finer granularity for ``hot'' parts of the sample space, where hot indicates a concentration of points.
The high-level strategy is to branch the decomposition at hot nodes in the hierarchy. 
In order to bound the memory allocation, we introduce a \textit{pruning} parameter $k$, which denotes the number of branches at each level of the hierarchy. 
Increasing $k$ allows for more branches and, thus, finer granularity at the cost of more memory.

In addition, as the space occupied by the generator is sublinear in the size of the database, we cannot rely on exact frequency counts to compute cardinalities for every \textit{possible} node in the decomposition.    
Therefore, we employ private sketches at deeper, and more populated, levels in the hierarchy to support approximate cardinality counting.
For example, at level $l$, a single private sketch can be used to count the number of points in each subdomain $\Omega_{\theta}$, for $\theta \in \{0,1\}^l$.
Once the private sketches process the data, they are then used to inform and grow the decomposition.
That is, nodes in the hierarchy are considered hot if their noisy approximate counts are large.
As the sketches are private, which means they are indistinguishable on neighboring inputs, the resulting decomposition is also private by the principle of post processing. 

To describe this process in more detail, we break it down into three components: initialization; parsing the data; and growing the partition. 

\begin{algorithm}[t]
    \SetAlgoLined
    \DontPrintSemicolon
    \SetKwProg{myproc}{define}{}{}
    \SetKwInput{Parameter}{Parameters}
    \SetKwInput{ds}{Data Structures}
    \KwIn{Database $\mathcal{X}$ of items from $\Omega$, pruning parameter $k$ and privacy parameters $\{\sigma_l\}$. }
    \myproc{$\construct(\mathcal{X}, (k, \lp, L), (w, j), \{\mathcal{D}_l\})$}
    {
        \tcp{\textsf{Initialize Data Structures}}
        Initialize complete binary partition tree $\tree$ with depth $\lp$ \label{line:init_one} \;
        \For{$l \in \{0,\ldots, \lp\}$}
        {
            \For{$v \in \{v \in \tree \mid {\level}(v) = l \}$}
            { 
                $ g \gets $ random value drawn from $\mathcal{D}_l$ \;
                $\nodecount \gets  g$ \label{line:exact_noise}\;
            }
        }
        \For{$l \in \{\lp+1,\ldots, L\}$}
        {
        $\sketch_l \gets$ initialize private sketch with dimension $(w,j)$ and noise $\mathcal{D}_l$\label{line:init_two}\;
        }
        \tcp{\textsf{Parse Dataset}}
        \For{$x \in \mathcal{X}$ } 
        {
            \For{$l \in \{0,\ldots, L\}$}
            {
                $\theta \gets $ the unique $\theta^{\prime} \in \{0,1\}^l$ such that $x \in \Omega_{\theta^{\prime}}$\;
                \If{$l\leq \lp$}
                {
                    Increment the counter in $\node \in \tree$\;
                }
                \Else 
                {
                    $\sketch_l.\update(\theta, 1)$ \;
                }  
            }
        }
        $\treehp \gets \grow(\tree, \{\sketch_l\}, k)$ \label{line:grow_part} \tcp*{Algorithm~\ref{alg:growpartition}}
        \KwRet $\treehp$  \;
    }
\caption{1-pass $\privhp$ algorithm.}
\label{alg:1passphd}
\end{algorithm}

\subsection{Initialization}

The boundaries of the subdomains $\subdomain$ can chosen arbitrarily.
However, these boundaries must be fixed a priori.
The pseudocode for initializing the component data structures is available in Algorithm~\ref{alg:1passphd} (Lines~\ref{line:init_one}-\ref{line:init_two}).
The decomposition of the domain $\Omega$ is encoded in a binary tree $\tree$, where each node in the tree $\node\in \tree$ represents the subset $\Omega_{\theta}$.
Let $\level(\node)$ represent the level to which $\node$ belongs.

The memory-utility trade-off for $\privhp$ is parameterized by $\lp$, the level at which pruning begins and $k$, the number of branches at each level $l>\lp$.
The initial decomposition contains all subsets $\Omega_{\theta}$, for $\theta \in \{0,1\}^{\leq \lp}$.
Therefore, the algorithm begins by initializing  $\tree$ as a \textit{complete} binary tree of depth $\lp$ (Line~\ref{line:init_one}). 

For a decomposition of depth $L$, the sampling distribution of the generator is based on the cardinalities of each subset $\Omega_{\theta}$ for $\theta \in \{0,1\}^{\leq L}$.
We store noisy \textit{exact} counts for subdomains $\subdomain$, where $\theta \in \{0,1\}^{\leq \lp}$, included in $\tree$ and noisy \textit{approximate} counts for subdomains $\subdomain$, where $\theta \in \{0,1\}^{\leq L} \setminus \{0,1\}^{\leq \lp}$.
The exact counters are stored in their corresponding nodes in $\tree$.
To ensure privacy, each counter at level $l\leq \lp$ is initialized with some random noise from a distribution $\mathcal{D}_l$ provided as input (Line~\ref{line:exact_noise}).
For each level $l> \lp$, the approximate counts for subsets $\Omega_{\theta}$, with $\theta \in \{0,1\}^l$, are stored in a private sketch ($\texttt{sketch}_l$) of dimension $w\times j$  and initialized with random noise from the  distribution $\mathcal{D}_l$ (Line~\ref{line:init_two}). 
These sketches will be used to grow the decomposition (Line~\ref{line:grow_part}) according to hot subsets of the sample space once the stream has been processed.

\subsection{Parsing the Data}

After initialization, we read the database in a single pass, one item at a time, while updating the internal data structures $\tree$ and $\{\sketch_l\}_{l>\lp}$.
For each update $x \in \mathcal{X}$, the procedure iterates through the levels in the hierarchy.
At each level $l$, the unique subdomain $\theta \in \{0,1\}^l$, such that $x\in \Omega_{\theta}$, is identified.
If $l\leq \lp$, then the counter in node $\node\in \tree$ is updated.
Otherwise, $\sketch_l$ is updated with the subset index $\theta$.
At the end of the stream, $\tree$ contains the noisy exact counts for subsets $\Omega_{\theta}$ with $\theta \in \{0,1\}^{\leq \lp}$ and the summaries $\{\sketch_l\}$ contain the noisy approximate counts for subsets at levels $l>\lp$.
At this point, these approximate counts are used to grow $\tree$ beyond level $\lp$.

\begin{algorithm}[t]
    \SetAlgoLined
    \DontPrintSemicolon
    \SetKwProg{myproc}{define}{}{}
    \SetKwInput{Parameter}{Parameters}
    \SetKwInput{ds}{Data Structures}
    \KwIn{Partition Tree $\tree$ and the collection of level-wise sketches $\{\sketch_l\}$ }
    \myproc{{$\grow(\tree, \{\sketch_l\}, k)$}}
    {
        Apply consistency to each non-leaf node $\node \in \tree$ in depth-first order using Algorithm~\ref{alg:consistency}\;
        $V \gets \{ \theta \mid \node \in \tree, \level(\node) = \lp\}$ \label{line:select_leaves}\;
        \For{$l \in \{ \lp+1, \ldots, L-1\}$}
        {
            \For{${\theta} \in V$}
            {
                \For{$\theta^* \in \{\theta 0, \theta 1\}$}
                {
                    $\hat{f}_{\theta^*} \gets \sketch_l.\query(\theta^*)$\;
                    add $v_{\theta^*}$ to $\tree$ with $v_{\theta^*}.\mathtt{count} = \hat{f}_{\theta^*}$\;
                }
                Apply consistency to $\node$\tcp*{Algorithm~\ref{alg:consistency}}
            }
            $V \gets $ the IDs of the Top-$k$ values in $\{\nodecount \mid v \in \tree, \level(v)=l+1\} $ \label{line:select_topk}\;
        }
        \KwRet $\tree$ \;
    }
\caption{Growing the $\privhp$ based on approximate counts at each level in the hierarchy.}
\label{alg:growpartition}
\end{algorithm}

\subsection{Growing the Partition}

Pseudocode for this step is available in Algorithm~\ref{alg:growpartition}, and an illustration of the process is provided in Figure~\ref{fig:grow_partition}.
Before growing the partition, a consistency step is preformed.
Consistency enforces two constraints.
First, it requires that the count of a parent node equals the sum of the counts of its child nodes.
Second, it requires that all counts are non-negative.
After processing the data, $\tree$ is not consistent due to the noise added for privacy.
The outcome of this consistency step is presented in Figure~\ref{fig:init_consistency}.
An equivalent consistency step is common in private histograms \cite{hay2009boosting}, where it is observed it can increase utility at the same privacy budget.

After the consistency step has been executed, the partition is expanded one level at a time.
The procedure begins by selecting the current leaf nodes $V$ of $\tree$ at level $\lp$ (Line~\ref{line:select_leaves}).
These are considered ``hot'' nodes.
Then, for each hot node $\node\in V$, it adds the two child nodes ($v_{\theta0}$ and $v_{\theta1}$) to $\tree$ as the decomposition of the node into two disjoint subsets (Figure~\ref{fig:gp_level3}).
In addition, the nodes $v_{\theta0}$ and $v_{\theta1}$ are initialized with the noisy frequency estimates retrieved from $\sketch_{\lp+1}$.
These estimates are then adjusted according to the consistency step (Figure~\ref{fig:consis_l3}).

After all the hot nodes have been expanded, the next iteration of hot nodes needs to be selected.
This is achieved by selecting the nodes with the Top-$k$ frequency estimates (Line~\ref{line:select_topk}).
With a new set of hot nodes, this process repeats itself and stops at depth $L-1$.
In summary, at each level in the iteration, the current hot nodes are branched into smaller subdomains at the next level in the hierarchy.
Then, the new subdomains with high frequency become hot at the subsequent iteration.

\begin{figure}[!ht]
    \centering
    \begin{subfigure}{0.31\textwidth}
        \includegraphics[width=\linewidth]{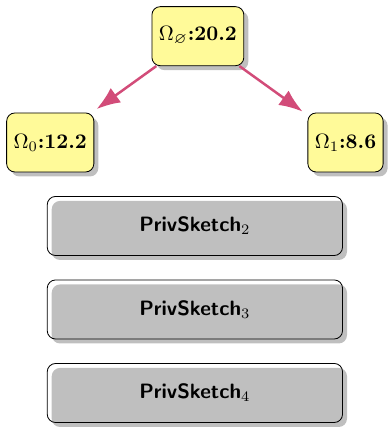}
        \caption{After processing the stream}
        \label{fig:input_growpartition}
    \end{subfigure}
    \hfill
    \begin{subfigure}{0.31\textwidth}
        \includegraphics[width=\linewidth]{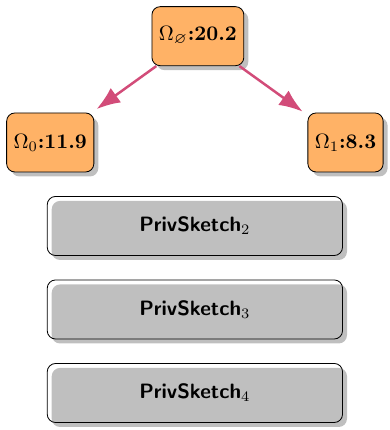}
        \caption{After consistency is applied to the tree}
        \label{fig:init_consistency}
    \end{subfigure}
    \hfill
    \begin{subfigure}{0.31\textwidth}
        \includegraphics[width=\linewidth]{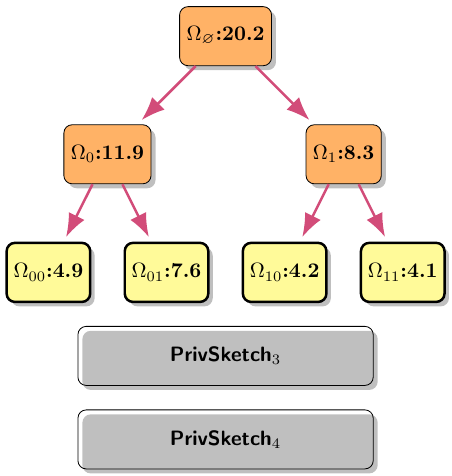}
        \caption{Adding nodes from $\sketch_2$}
        \label{fig:gp_level3}
    \end{subfigure}

    \vspace{10pt} 
    \begin{subfigure}{0.31\textwidth}
        \includegraphics[width=\linewidth]{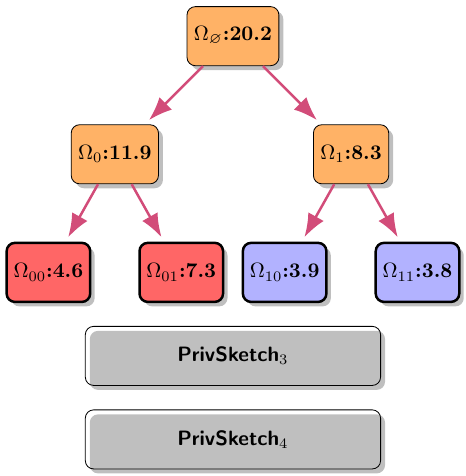}
        \caption{After consistency is applied to level 2}
        \label{fig:consis_l3}
    \end{subfigure}
    \hfill
    \begin{subfigure}{0.31\textwidth}
        \includegraphics[width=\linewidth]{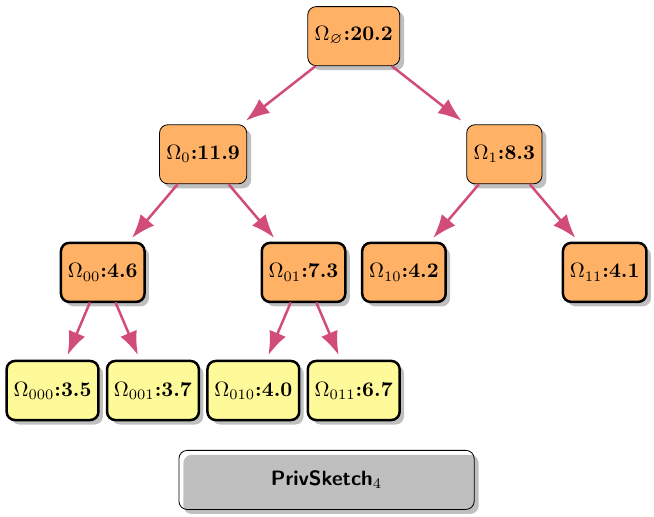}
        \caption{Pruning based on top-$2$ selection}
        \label{fig:pruning}
    \end{subfigure}
    \hfill
    \begin{subfigure}{0.31\textwidth}
        \includegraphics[width=\linewidth]{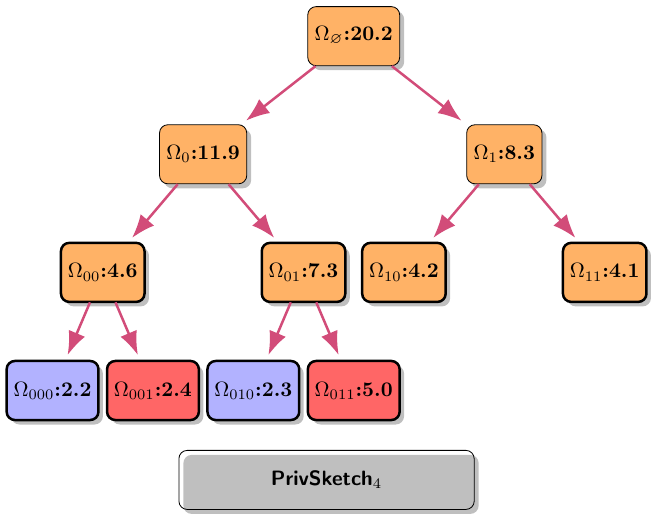}
        \caption{After consistency is applied to level 3}
        \label{fig:output_growpartiton}
    \end{subfigure}

    \caption{Illustration of Algorithm~\ref{alg:growpartition} with $k=2, \lp = 1$ and $L=4$.
    Figure~\ref{fig:input_growpartition} represents its input.}
    \label{fig:grow_partition}
\end{figure}

\subsection{Consistency}

Consistency ensures that (1) all counts are non-negative and (2) that the counts of two subregions add to the count of their parent region. 
The consistency step in Algorithm~\ref{alg:growpartition} is general and, following He \textit{et al.}~\cite{he2023algorithmically}, adheres to the following rule.
In the case of a deficit, when the
sum of the two subregional counts is smaller than the count of the parent region, both subregional counts should be increased. 
Conversely, in the case of a surplus, both subregional counts should be decreased. 
Apart from this requirement, we are free to distribute the deficit or surplus between the subregional counts.

Our main utility bound (Theorem~\ref{thm:utilityHP}) is based on a {concrete} instance of consistency.
The method we adopt is presented in Algorithm~\ref{alg:consistency}.
The main idea is to \textit{evenly} redistribute the error generated from sibling subregions. 
To formalize this idea, let $\node$ denote the parent node.
First, we calculate the difference between the subregional counts and their parent region (Line~\ref{line:sum}): $\Lambda = v_{\theta 0}.\counter + v_{\theta 1}.\counter - \node.\counter$.
This difference is then evenly redistributed across the subregions (Line~\ref{line:consistency}):
\begin{align}
    v_{\theta 0}.\counter \gets v_{\theta 0}.\counter - \Lambda/2 \quad \text{and} \quad v_{\theta1}.\counter \gets v_{\theta1}.\counter - \Lambda/2.
    \label{eqn:redistribution}
\end{align}
We also add two correction steps for whenever this approach might violate consistency.
The first correction makes sure that both subregional counts are non-negative prior to applying consistency (Line~\ref{line:correction1}).
The second involves applying a different redistribution method in the event that~\eqref{eqn:redistribution} violates consistency (Line~\ref{line:correction21}).
In this instance, the count of the violating node is set to $0$ and its sibling node inherits the full count from its parent.
Both correction steps \textit{reduce} the amount of error in the subregion counts.

\begin{algorithm}[t]
    \SetAlgoLined
    \DontPrintSemicolon
    \SetKwProg{myproc}{define}{}{}
    \SetKwInput{Parameter}{Parameters}
    \SetKwInput{ds}{Data Structures}
    \KwIn{Node $\node\in \tree$}
    \myproc{${\enforce}(\node)$}
    {
        \For{$\theta^* \in \{\theta 0, \theta 1\}$}
        {
            \If{$v_{\theta^*}.\counter <0$ \label{line:correction1}} 
            {
                \tcp{Error Correction Type 1}
                $v_{\theta^*}.\counter \gets 0$\;
            }
        }
        $\Lambda \gets v_{\theta 0}.\counter + v_{\theta 1}.\counter - \node.\counter$\;
        \label{line:sum}
        \If{$\min \{v_{\theta 0}.\counter - \Lambda/2, v_{\theta 1}.\counter - \Lambda/2\} < 0$ \label{line:correction21}}
        {
            \tcp{Error Correction Type 2}
            $(\theta_{\mathtt{min}}, \theta_{\mathtt{max}}) \gets$ order $(\theta_0, \theta_1)$ by their counters\;
            $v_{\theta_{\mathtt{min}}} \gets 0$\;
            $v_{\theta_{\mathtt{max}}}.\counter \gets \node.\counter$\;
        }
        \Else
        {
            \For{$\theta^* \in \{\theta 0, \theta 1\}$}
            {
            $v_{\theta^*}.\counter \gets v_{\theta^*}.\counter - \Lambda/2  $\;
            \label{line:consistency}
            }
        }
        \KwRet \;
    }
\caption{Enforcing consistency between nodes in $\tree$.}
\label{alg:consistency}
\end{algorithm}

\section{Private Synthetic Data}
\label{sec:private_synthetic}
An item can be sampled from the decomposition tree $\tree$ by selecting a number $u$ uniformly in the range $[0, v_{\varnothing}.\mathtt{count}]$, where $v_{\varnothing}$ is the root node of $\tree$.
Then, a root-to-leaf traversal of the tree is performed.
At each node $\node$ on the path, we retrieve the count from the left child $c\gets v_{\theta 0}.\mathtt{count}$.
We branch left if $c \geq u$; otherwise, we branch right.
When branching right, $u$ is updated with $u\gets u -\node.\counter$.
The final leaf node represents a subset of the sample space and we can return any item uniformly at random from this subset.

Note that this sampling algorithm can take any binary decomposition of $\Omega$ as input.
This makes any tree $\tree$ synonymous with a sampling distribution.
Therefore, throughout the rest of the paper we often refer to $\tree$ as a probability distribution.
For the remainder of this section, we establish privacy and provide bounds on the utility of the generator $\treehp$ output by Algorithm~\ref{alg:1passphd}, where utility is measured in the expected 1-Wasserstein metric.

\subsection{Privacy}

Algorithm~\ref{alg:growpartition} is completely deterministic.
Therefore, if the inputs to Algorithm~\ref{alg:growpartition} are differentially private, then the resulting partition is differentially private by the post-processing property (Lemma~\ref{lem:post_processing}).
The random perturbations introduced at initialization depend on the collection of noise distributions $\{\mathcal{D}_l\}$.
They should provide sufficient noise such that the output distributions of $\tree$ and $\{\sketch_l\}$ on neighboring datasets are indistinguishable.
The are many choices for $\{\mathcal{D}_l\}$ that impact both privacy and utility.
Here is one example.
\begin{restatable}{theorem}{privacythm}
    If the noise distributions have the following form:
    \begin{align}
        \mathcal{D}_l &= 
            \begin{cases}
                    {\laplacedist}(\sigma_l^{-1})  &  \text{for } l\leq \lp \\
                    {\laplacedist}^{w\times j}(j\sigma_l^{-1}) & \text{Otherwise }
            \end{cases}\;
    \label{eqn:laplace_parameters}
    \end{align}
    Then, the decomposition $\treehp$ output by Algorithm~\ref{alg:1passphd} is $\varepsilon$-differentially private for $\sum_{l=0}^{L}\sigma_l = \varepsilon$. 
    \label{thm:dp_result}
\end{restatable}
\begin{proof}
    On neighboring datasets $X= X^{\prime} \cup \{x\}$, we are required to minimize the effect of the additional element $x$ on the output distribution of the process.
    During data processing, the sensitive element $x$ impacts both the initial\footnote{By initial partition tree, we refer to the tree prior to the growing phase that occurs after data processing.} partition tree $\tree$ and the sketches $\{\sketch_l\}$.
    We consider both cases separately.
    With $\tree$ we store (noisy) exact counts.
    The sensitive element $x$ traverses a single root to leaf path, updating each node on the path, when it is processed.
    The counts are incremented by 1 and the path has length $\lp$.
    Therefore, the sensitivity of the initial partition tree is $\lp$.
    With noise $\laplacedist(\sigma_l^{-1})$ applied to each count on the path, the initial partition tree is $\sum_{l=0}^{\lp}\sigma_l$-differentially private by Lemmas~\ref{lem:laplace_mechanism} \& ~\ref{lem:composition}.
    As previously noted, a sketch has sensitivity $j$. 
    Therefore, $\laplacedist^{w\times j}(j\sigma_l^{-1})$ noise provides $\sigma_l$-differential privacy for the sketch at level $l$.
    Basic composition (Lemma~\ref{lem:composition}) and the observation that there are $(L-\lp)$ sketches completes the proof.
\end{proof}

\subsection{Utility and Performance}

We begin by introducing some notation.
Let $\gamma_l = \max_{\theta \in \{0,1\}^l} \diam(\subdomain)$ and $\Gamma_l = \sum_{\theta \in \{0,1\}^l}\diam(\subdomain).$
Let $C_l= \langle |\Omega_0|, \ldots, |\Omega_{2^l}| \rangle$ denote the vector of subdomain cardinalities at level $l$. 
To help capture the effect of pruning, we use the vector $\tail_k^l$ to denote $C_l$ with the top-$k$ cardinalities set to 0. 
For skewed inputs $||\tail_k^l||_1$ is small and can even be $0$ for sparse inputs.

Algorithm~\ref{alg:1passphd} is general and doesn't prescribe the type of private sketch or the noise distributions of the perturbations. 
For our concrete results, we use a private Count-min Sketch as the sketching primitive and follow the noise distributions of Lemma~\ref{thm:dp_result}.
\begin{theorem}
    On input $\mathcal{X}$, with sketch dimensions $(w=2k, j)$ and partition dimensions of pruning level $\lp$, hierarchy depth $L$ and pruning parameter $k$,
    for $\varepsilon = \sum_{l=0}^L \sigma_l$,
    Algorithm~\ref{alg:1passphd} produces a partition $\treehp$ that is $\varepsilon$-differentially private and has the following distance from the empirical distribution $\mu_{\mathcal{X}}$ in the expected 1-Wasserstein metric:
    \begin{align}
        \mathbb{E}[W_1(\mu_{\mathcal{X}},\treehp)] 
        &= 
        \noise + \approximation 
        \label{eqn:utility_bound}
    \end{align}
    where
    \begin{align*}
        \noise = \mathcal{O}\left(\frac{1}{n}\left(\sum_{l=0}^{\lp} \frac{\Gamma_{l-1}}{\sigma_l} + \sum_{l=\lp+1}^{L} \frac{kj\gamma_{l-1}}{\sigma_l}\right)\right), \quad
        \approximation = \mathcal{O}\left(\left(\frac{||\tail_k^L||_1}{n} + 2^{-j}\right)  \sum_{l=\lp+1}^{L} \gamma_{l-1}\right)
    \end{align*}
    \label{thm:utilityHP}
\end{theorem}
\noindent
The proof of this result is the content of Section~\ref{sec:proof}.
The components allow us to make sense of the bound.
The $\noise$ term represents the distance incurred, between $\mu_{\mathcal{X}}$ and $\treehp$, due to noise added for privacy.
This noise affects both the counts and the pruning procedure.
The $\approximation$ term represents the reduction in utility due to approximation. 
The $||\tail_k^L||_1$ term is dependent on the underlying distribution of the input $\mathcal{X}$.

The privacy and accuracy guarantees of Theorems~\ref{thm:dp_result} and~\ref{thm:utilityHP} hold for any choice of $\{\sigma_l\}$.
By optimizing the $\{\sigma_l\}$, we can achieve the best utility for a given level of privacy $\varepsilon = \sum_{l=0}^{L} \sigma_l$.
\begin{restatable}{lemma}{thmoptacc}
    With the optimal choice of privacy parameters, on input $\mathcal{X}$ and partition dimensions of $(k, \lp, L)$, the loss in utility due to noise perturbations ($\noise$ in \eqref{eqn:utility_bound}) is:
    \begin{align*}
        \noise = \mathcal{O}\left(\frac{1}{\varepsilon n} \left( \sum_{l=0}^{\lp}  \sqrt{\Gamma_{l-1}} + \sum_{l=\lp+1}^L \sqrt{j k\gamma_{l-1}} \right)^2 \right)
    \end{align*}
    \label{lem:opt_acc}
\end{restatable}
\noindent  
Due to its similarity to Theorem 11 in He \textit{et al} \cite{he2023algorithmically}, the proof is relegated to Appendix~\ref{sec:app_opt_sigma}.
Theorem~\ref{thm:utilityHP} applies for any input domain $\Omega$.
To make the result more tangible and to demonstrate its applicability, following prior work \cite{boedihardjo2024private,he2023algorithmically}, we apply it (in conjunction with Lemma~\ref{lem:opt_acc}) to the hypercube $\Omega=[0,1]^d$.
This leads to the following result.

\begin{restatable}{corollary}{corollaryhypercube}
    When $\Omega = [0,1]^d$ equipped with the $l^{\infty}$ metric, for pruning parameter $k$, $\privhp$ can process a stream $\mathcal{X}$ of size $n$ in $M=\mathcal{O}(k\log^2(n))$ memory and $\mathcal{O}(\log (\varepsilon n))$ update time.
    $\privhp$ can subsequently output a $\varepsilon$-differentially private synthetic data generator $\treehp$, in $\mathcal{O}(M \log n)$ time, such that
    \begin{align*}
        \mathbb{E}[W_1(\mu_{\mathcal{X}}, \treehp )] = \begin{cases}
             \mathcal{O}\left( \frac{\log^2(M)}{\varepsilon n} + \frac{||\tail_k^{\varepsilon n}||}{Mn} \right) & if d = 1  
            \\
            \mathcal{O}\left(\frac{M^{(1-\frac{1}{d})}}{\varepsilon n}+ \frac{||\tail_k^{\varepsilon n}||}{M^{1/d}n}\right) & if d\geq 2 
    \end{cases},
    \end{align*}
    \label{cor:hypercube}
\end{restatable}
\noindent
The proof is the content of Section~\ref{sec:proof_hypercube}. 
For comparison, the state-of-the-art in the static setting, $\pmm$ \cite{he2023algorithmically}, achieves a utility bound of $\mathcal{O}(\log^2(\varepsilon n)/ (\varepsilon n))$, for $d=1$, with a memory allocation of $\mathcal{O}(\varepsilon n)$ (See Table~\ref{tab:performance_results}). 
Thus, through the hierarchy pruning parameter $k$, $\privhp$ provides a smooth interpolation from the optimal static case to a memory bounded environment.
The same observation is true for $d=2$. 
Further, for sparse or highly skewed inputs, where $||\tail_k^{\varepsilon n}||_1$ is small, pruning may even improve the utility bound, as fewer nodes in the hierarchy results in less noise being added.

\section{Measuring Utility}
\label{sec:m_utility}
Before proving Theorem~\ref{thm:utilityHP}, we need a method to quantify the distance between $\empiricalx$ and $\treehp$.  
In the empirical distribution, each point \( x \in \mathcal{X} \) carries a unit of probability mass. 
When a point is abstracted into a set within a partition representing a generator, its probability mass is evenly distributed across the set. 
This reflects the process where, conditioned on a set being selected by the generator, a synthetic point is uniformly sampled from the set. 
The total distance this probability mass moves during abstraction is bounded by the diameter of the subdomain.  
Similarly, modifications to node counts in the decomposition tree result in shifts of probability mass within the generator.
Bounding the utility of the generator, therefore, involves constraining the distance these probability masses move as $\empiricalx$ transforms into $\treehp$.  

To formalize these bounds, we introduce new terminology that captures errors arising from both noise perturbations and frequency approximations. 
This terminology also enables a precise analysis of how the consistency step balances these errors across nodes. 
The consistency step adjusts a parent node’s count by redistributing it among its child nodes. 
Inaccuracies in this redistribution, which we refer to as a \textit{consistency error}, maintain a divergence between the empirical distribution and the synthetic data generator. 
The utility cost of a consistency error depends on its magnitude and the size of the affected subdomain. 
Consequently, bounding the utility loss from noise and approximation requires both measuring each consistency error and identifying its location.

\subsection{Quantifying a Consistency Error}
\label{app:quantify_miss}

A consistency error represents the transfer of probability mass from one subdomain to its sibling subdomain, altering the probability distribution encoded by the underlying decomposition.
To arrive at a formal expression of a consistency error, we begin by introducing some notation.
Let $c_{\theta} = |\subdomain|$ denote the exact count at node $v_{\theta}$, let $\lambda_{\theta}$ denote the noise added to $\node.\counter$ from privacy perturbations, and let $e_{\theta}$ denote the approximation error added to $v_{\theta}.\textsf{count}$ due to hashing collisions in the sketch.
Lastly, we define $\miss(\node)$ as the size of the consistency error incurred at node $\node$.

To help quantify $\miss(\node)$, we take an accounting approach, where noisy approximate counts are disaggregated into various components using the notation introduced above.
This allows us to identify which part of the adjusted consistent counts constitutes an error.
As we do not want to double count a consistency error, $\miss(\node)$ does not include consistency errors that occur at ancestor nodes and are, subsequently, inherited, due to previous consistency steps, in the count at $\node$. 
Therefore, the size of $\miss(\node)$ is solely influenced by the errors in its two child nodes.

\begin{examplebox}
This example evaluates $\miss(\node)$ for the subtree in Figure~\ref{fig:miss_example}.
The cardinalities of the subdomains are $c_{\theta} = 5, c_{\theta0}=3, c_{\theta 1}=2$.
From~\eqref{eqn:parent_after}, it follows that $\err_{\theta}=\node.\counter^{\mathtt{after}}-c_{\theta}=-0.4$.
The component errors in the child nodes are $\lambda_{\theta 0}=-0.5,e_{\theta 0}=1, \lambda_{\theta 1}=-0.3, e_{\theta 1}=2$.
Thus, the child counts prior to consistency are
\begin{align*}
    v_{\theta0}.\counter^{\mathtt{before}} &= c_{\theta0} + \lambda_{\theta 0} + e_{\theta0} = 3.5 
    &v_{\theta1}.\counter^{\mathtt{before}} = c_{\theta1} + \lambda_{\theta 1} + e_{\theta1} = 3.7.
\end{align*}
Using the formula in~\eqref{eqn:miss_definition}, the size of the $\miss$ at $\node$ is
\[
 \miss(v_{\theta}) = |(\lambda_{\theta0}- \lambda_{\theta1} + e_{\theta0}-e_{\theta1})/2| = 0.6
\]
This value can be expressed as a portion of the consistent counts in the child nodes.
\begin{align*}
    v_{\theta0}.\counter^{\mathtt{after}} &= v_{\theta0}.\counter^{\text{before}} - \Lambda/2 = c_{\theta 0} +\frac{\err_{\theta}}{2} - \miss(\node) = 2.2 \\
    v_{\theta1}.\counter^{\mathtt{after}} &= v_{\theta1}.\counter^{\text{before}} - \Lambda/2 = c_{\theta 1} +\frac{\err_{\theta}}{2} + \miss(\node) = 2.4
\end{align*}
Therefore, a $\miss$ can be understood as the count transferred from one subdomain to its sibling \textit{after} the subdomain cardinalities have been adjusted by the existing error in the parent node.
For example, the total error in $v_{\theta 0}.\counter$ after consistency is 0.8.
However, 0.2 of this error comes from consistency errors at ancestor subdomains.
Therefore $\miss(\node)$ captures the precise count transferred between $\Omega_{\theta 0}$ and $\Omega_{\theta 1}$ due to local errors.
\label{example:miss}
\end{examplebox}
\begin{figure}[htbp]
    \centering
    \begin{subfigure}{0.35\textwidth} 
        \centering
        \includegraphics[width=\textwidth]{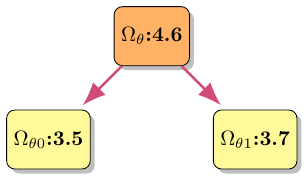} 
        \caption{Before consistency}
        \label{fig:sub1}
    \end{subfigure}
    \hfill 
    \begin{subfigure}{0.35\textwidth}
        \centering
        \includegraphics[width=\textwidth]{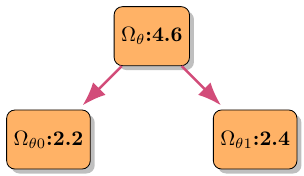} 
        \caption{After consistency}
        \label{fig:sub2}
    \end{subfigure}
    \caption{Subtree used for Example~\ref{example:miss}}
    \label{fig:miss_example}
\end{figure}
A consistency error is dependent on the method employed for consistency. 
The method we adopt (Algorithm~\ref{alg:consistency}) has the following form. 
\begin{align}
    \Lambda = v_{\theta 0}.\counter + v_{\theta 1}.\counter  
    &- 
    \node.\counter. \label{eqn:cons1} \\
    v_{\theta 0}.\counter \gets v_{\theta 0}.\counter - \Lambda/2 \quad 
    &\text{and} 
    \quad v_{\theta1}.\counter \gets v_{\theta1}.\counter - \Lambda/2.
    \label{eqn:cons2}
\end{align}

Note that Algorithm~\ref{alg:consistency} also contains two correction steps for when~\eqref{eqn:cons2} might violate consistency.
We will address these correction steps, and how they affect $\miss$, at a later point.

For clarity, let $\counter^{\text{before}}$ refer to a count before consistency is applied and $\counter^{\text{after}}$ refer to a count after it is made consistent.
When consistency is enforced at $v_{\theta}$, consistency has already been applied at the parent of $v_{\theta}$.
Therefore, the following equality already holds:
\begin{align}
    v_{\theta}.\counter^{\text{after}} = c_{\theta} + \err_{\theta},
    \label{eqn:parent_after}
\end{align}
where $\err_{\theta}$ accumulates consistency errors inherited at $v_{\theta}$ from all its ancestors.
Prior to consistency, the count in a child node $v_{\theta0}$ has the following form:
\begin{align*}
    v_{\theta0}.\counter^{\text{before}} 
    &=
    c_{\theta0} + \lambda_{\theta 0} + e_{\theta0}.
\end{align*}
Note that $e_{\theta0}=0$ if $\level(v_{\theta0})\leq \lp$, as no sketches are used.
To enforce consistency between child nodes, an adjustment variable (See~\eqref{eqn:cons1}) is calculated:
\begin{align*}
    \Lambda &= v_{\theta0}.\counter^{\text{before}} + v_{\theta1}.\counter^{\text{before}} - v_{\theta}.\counter^{\text{after}} \\ 
    &= c_{\theta0} + \lambda_{\theta 0} + e_{\theta0} +c_{\theta1} + \lambda_{\theta 1} + e_{\theta1} - c_{\theta} - \err_{\theta} \\
    &=  \lambda_{\theta 0} + e_{\theta0} + \lambda_{\theta 1} + e_{\theta1}  - \err_{\theta}
\end{align*}
Focusing on the left child $v_{\theta 0}$,
the size of $\miss$ can be inferred by calculating $v_{\theta 0}.\counter^{\text{after}}$ (See~\eqref{eqn:cons2}).
\begin{align}
    v_{\theta 0}.\counter^{\text{after}} &= v_{\theta 0}.\counter^{\text{before}} - \Lambda/2 \nonumber \\ 
    &=  c_{\theta0} + \lambda_{\theta 0} + e_{\theta0} - (\lambda_{\theta 0} + e_{\theta0} + \lambda_{\theta 1} + e_{\theta1}  - \err_{\theta})/2 \nonumber  \\
    &=c_{\theta0} + (\lambda_{\theta0}- \lambda_{\theta1} + e_{\theta0}-e_{\theta1})/2 +\err_{\theta}/2
    \label{eqn:misses_consistency}
\end{align}
As $\err_{\theta}$ refers to points already counted as errors (which we do not wish to count twice),
it follows that
\begin{align}
    \miss(v_{\theta}) &= |(\lambda_{\theta0}- \lambda_{\theta1} + e_{\theta0}-e_{\theta1})/2|.
    \label{eqn:miss_definition}
\end{align}
Therefore, a consistency error occurs when there is a difference in the errors between sibling nodes. 
This difference is evenly split between the subdomains.
For greater clarity, concrete example of a consistency error is provided in Example~\ref{example:miss}.
With~\eqref{eqn:miss_definition} in place, we can now pursue a bound for $\mathrm{E}[\miss(\node)]$.
\begin{restatable}{lemma}{lemmiss}
    With consistency applied according to Algorithm~\ref{alg:consistency}, on sketch parameters $2w$ and $j$ and the noise distribution from Equation~\ref{eqn:laplace_parameters}, then for all internal nodes $\node\in \treehp$:
    \begin{align*}
        \mathbb{E}[{\miss}(\node)] \leq 
        \begin{cases}
            2\sqrt{2}\sigma_{l+1}^{-1} & \level(\node ) < \lp \\
            2\sqrt{2}\sigma_{l+1}^{-1} \cdot j + \frac{||\textup{\textsf{tail}}_w^{l+1}||}{w}  + 2^{-j+1} n &   \textup{\text{Otherwise}}
        \end{cases}
    \end{align*}
    \label{lem:miss_error}
\end{restatable}

\section{Proof of Theorem~\ref{thm:utilityHP}}
\label{sec:proof}

We break down the proof into a series of steps that are \textit{equivalent} to Algorithm~\ref{alg:1passphd}.
$\treehp$ can be constructed from ${\mathcal{X}}$ using the following steps: 
\begin{enumerate}[label=\texttt{Step (\arabic*)}, leftmargin = 1.5cm]
    \item \label{step:two}
    Construct a partition tree $\mathcal{T}_{\mathcal{X}}$ that summarizes $\mathcal{X}$ using \textit{exact} counts and a \textit{complete} binary hierarchical partition of depth $L$ (Figure~\ref{fig:proof_step1}).
    Then conduct \textit{exact} pruning on $\mathcal{T}_{\mathcal{X}}$ by branching at the nodes with the exact top-$k$ counts at each level $l \geq \lp$ to produce $\treetop$.
    An example of $\treetop$ is provided in Figure~\ref{fig:proof_step2}. 
    \item \label{step:three}
    Constructing $\treeprune$ by adjusting $\treetop$ so that its structure matches $\treehp$.
    $\treeprune$ captures the effect of approximate pruning.
    Note that $\treeprune$ still has exact counts.
    An example of  $\treeprune$ is provided in Figure~\ref{fig:proof_step3}.
    \item \label{step:four}
    Add the privacy noise and approximation errors to the exact counts in $\treeprune$ and apply the consistency step to produce $\treehp$.
    An example of $\treehp$ is provided in Figure~\ref{fig:proof_step4}.
\end{enumerate}

\begin{figure}[]
    \centering
    \begin{subfigure}[t]{0.45\textwidth}
        \centering
        \includegraphics[width=\textwidth]{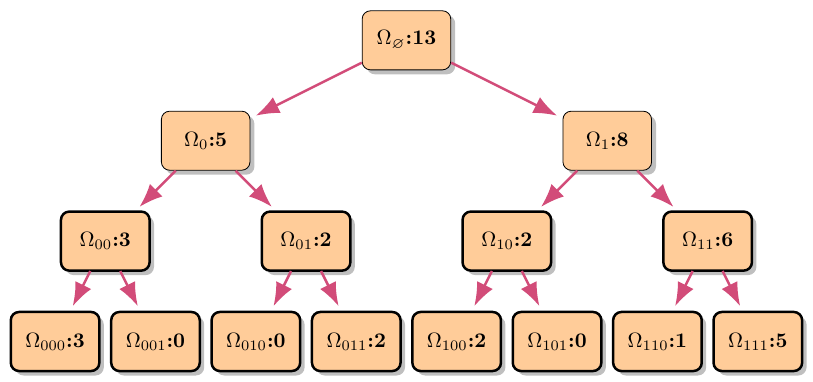}
        \caption{
        Example of $\treex$; a complete binary hierarchical decomposition of depth $L$.
        }
        \label{fig:proof_step1}
    \end{subfigure}
    \hfill
    \begin{subfigure}[t]{0.45\textwidth}
        \centering
        \includegraphics[width=\textwidth]{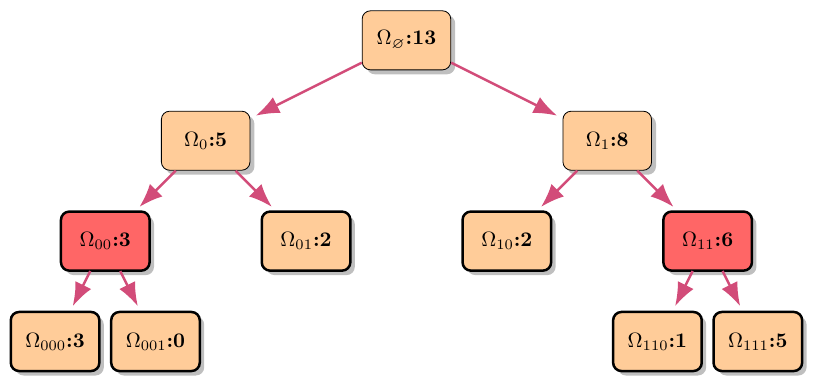}
        \caption{
        Example of $\treetop$.
        The distance between $\empiricalx$ and $\treetop$ is bound in Lemma~\ref{lem:step2}.
        }
        \label{fig:proof_step2}
    \end{subfigure}

    \vspace{0.8em} 

    \begin{subfigure}[t]{0.45\textwidth}
        \centering
        \includegraphics[width=\textwidth]{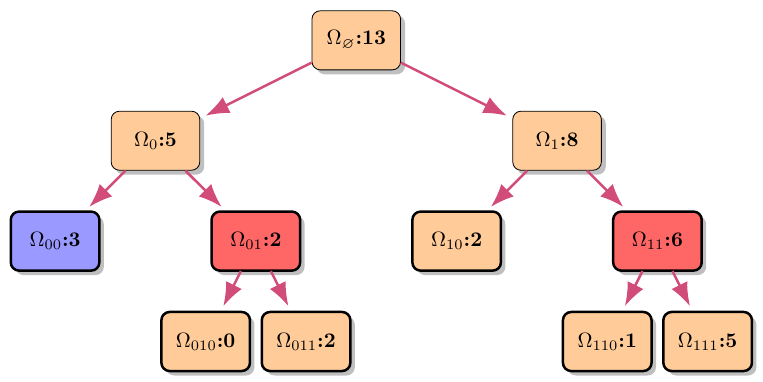}
        \caption{
        Example of $\treeprune$.
        The distance between $\treetop$ and $\treeprune$ is bound in Lemma~\ref{lem:step3}.
        }
        \label{fig:proof_step3}
    \end{subfigure}
    \hfill
    \begin{subfigure}[t]{0.45\textwidth}
        \centering
        \includegraphics[width=\textwidth]{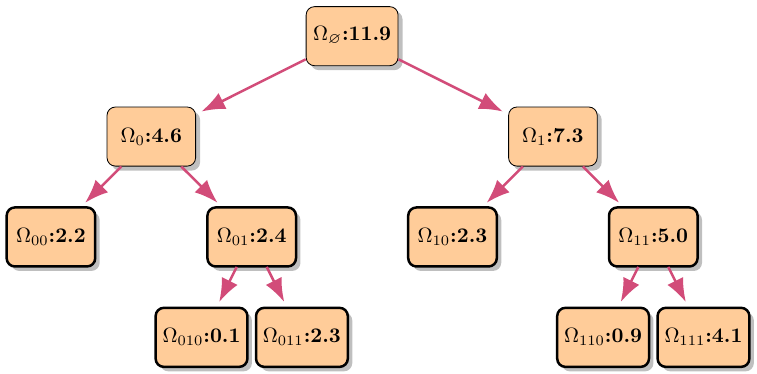}
        \caption{
        Example of $\treehp$.
        The distance between $\treehp$ and $\treeprune$ is bound in Lemma~\ref{lem:step4}.
        }
        \label{fig:proof_step4}
    \end{subfigure}

    \caption{Illustration of the proof pipeline for Theorem~\ref{thm:utilityHP}.
    $k = 2$, $\lp = 2$, $L=3$.
    }
    \label{fig:proof_pipeline}
\end{figure}

Figure~\ref{fig:proof_pipeline} illustrates the sequence of trees produced by this process.
Note that these steps are {equivalent} to Algorithm~\ref{alg:1passphd}.
However, they are not the same and are only introduced for analytic purposes.
By the triangle inequality, it follows that, to bound $\mathbb{E}[W_1(\empiricalx,\treehp)]$, it suffices to bound the distance between each pair of trees in the sequence.
We now bound the cost of each step separately, beginning with~\ref{step:two}.

At \ref{step:two}, exact pruning has the effect of merging sparse sibling leaf nodes into larger subdomains.
This reduces the granularity of the partition and, thus, its utility.
The exact pruning step allows us to quantify the impact of the underlying data distribution on the utility of the data sampler.
For example, distributions that are highly skewed will maintain a majority of their points in top-$k$ nodes. 
Therefore, pruning will have a smaller impact on the utility of the sampler under high skew. 
This loss in utility is bounded by the following result.
\begin{lemma}
    $\mathbb{E} W_1(\empiricalx, \treetop) \leq \frac{1}{n} ||\tail_k^L||_1 \cdot \sum_{l=\lp+1}^{L-1} \gamma_l$
    \label{lem:step2}
\end{lemma}
\begin{proof}
    Given the complete binary decomposition of $\mathcal{X}$, the procedure to construct $\treetop$ iterates from level $\lp+1$ to the bottom of the tree, selecting to keep the branches with the Top-$k$ counts. 
    Let $\theta_l^{(i)}$ name the candidate node with the  $i^{\text{th}}$ largest exact cardinality at level $l$.
    If a node $\node$ is pruned, then each point $x\in \subdomain \cap \mathcal{X} $ is abstracted into subdomain $\subdomain$. 
    This is equivalent to redistributing the probability mass of each point in $x\in \subdomain \cap \mathcal{X}$ by moving it a distance of at most $\diam(\subdomain)$. 
    As there are $2k$ candidates at each level of pruning, it follows that
    \begin{align*}
        W_1(\empiricalx, \treetop)
        &\leq 
        \frac{1}{n}\sum_{l=\lp+1}^{L-1} \sum_{x=k+1}^{2k} |\Omega_{\theta_{(x)}}|\cdot\diam(\Omega_{\theta_{(x)}}) 
        \leq \frac{1}{n} \sum_{l=\lp+1}^{L-1} || \tail_k^l||_1\gamma_l
        \leq \frac{|| \tail_k^L||_1}{n}  \cdot \sum_{l=\lp+1}^{L-1} \gamma_l
    \end{align*}
    The final inequality is based on the observation that $|| \tail_k^{l-1}||_1 \leq || \tail_k^l||_1$, as the subdomains that are top-$k$ in level $l-1$ are split into $2k$ subdomains at level $l$, which now compete for top-$k$ membership.
\end{proof}

The remaining steps handle the consistency errors. 
When pruning according to the (noisy and approximate) consistent counts, each consistency error propagates down the tree and affects future pruning decisions.
The main challenge in the proof of the bound for \ref{step:three} involves demonstrating that the influence of each consistency error \textit{decays} as it propagates down the tree.
The following result bounds this utility loss.
\begin{restatable}{lemma}{lemstepthree}
     $\mathbb{E} W_1(\treetop, \treeprune) \lesssim \frac{4}{n}\left(\sum_{l=1}^{\lp} 2\sqrt{2}\frac{\Gamma_{l-1}}{\sigma_l}  + \sum_{l=\lp+1}^{L} (2\sqrt{2}\frac{jk}{\sigma_l}  + {||\tail_k^l||_1} + \frac{n}{2^{j-1}} ) \gamma_{l-1}\right)$
    \label{lem:step3}
\end{restatable}
\begin{proof}
    Set sketch width $w=2k$.
    At \ref{step:three}, we adjust the tree structure to reflect the impact of noisy approximate pruning (Figure~\ref{fig:proof_step3}).
    Nodes can ``jump'' into the top-$k$ due to the influence of consistency errors.
    We refer to the \textit{size} of a jump as the difference in exact counts between the jumping node and the true top-$k$ node it displaces. 
    A jump of size $c$ occurs, for some $\theta_1, \theta_2 \in \{0,1\}^{l}$, with $l>\lp$, under the following conditions:
    \begin{itemize}
        \item $v_{\theta_1}.\counter> v_{\theta_2}.\counter$;
        \item $|\Omega_{\theta_1}|+c = |\Omega_{\theta_2}|$ for integer constant $c>0$;
        \item and, when $v_{\theta_2}$ is in the exact top-$k$ and $v_{\theta_1}$ is \textit{not} in the exact top-$k$.
    \end{itemize}
    The \textit{cost} of each jump is at most $c \cdot \gamma_{l}$, representing $c$ additional points that are abstracted into a subdomain at level $l$.
    Therefore, the distance between the sampling distributions defined by $\treetop$ and $\treeprune$ is determined by the size of the jumps that occur due to approximate pruning.
    Further, the size of the jump is bound by the total number of consistency errors that are inherited from ancestors at that node.  
    
    A consistency error can cascade through hot nodes, affecting pruning decisions across multiple levels.
    To formally capture the notion of a specific consistency error propagating down the hierarchy, let $\cost_l(\node)$ represent the utility cost incurred at level $l$ due to the consistency error \textit{originating} at $\node$.
    For example, for a hot node $\node$, at level $l_{\theta} = \level(\node)$, 
    the error $\miss(\node)$ is passed to both candidate subdomains $v_{\theta 0}$ and $v_{\theta 1}$.
    One subdomain will receive $\miss(\node)$ as an addition, allowing it to jump upwards, and the other subdomain will receive $\miss(\node)$ as a subtraction, allowing it to jump downwards.
    Therefore, the utility cost incurred by the consistency error of node $\node$ at level $l_{\theta}+1$ is
    \begin{align*}
        \cost_{l_{\theta}+1}(\node) &\leq \miss(\node) \cdot \gamma_{l_{\theta}+1} + \miss(\node)\cdot \gamma_{l_{\theta}+1} = 2\miss(\node) \cdot \gamma_{l_{\theta}+1}.
    \end{align*}
    If these subdomains are both hot, then, due to consistency, $\miss(\node)$ is propagated to the next level and is evenly redistributed between the child nodes\footnote{See Equation~\eqref{eqn:misses_consistency}, where errors from the parent node are split evenly in the consistent count of the child node.}.
    Thus, $v_{\theta 00}, v_{\theta 01}, v_{\theta 10}, v_{\theta 11}$ contain (as either an addition or subtraction) $\miss(\node)/2$ in their counts.
    If either subdomain is cold, then its portion of $\miss(\node)$ ceases to participate in pruning decisions.
    Assuming the worst-case, where both subdomains are hot,
    \begin{align*}
        \cost_{l_{\theta}+2}(\node) &\leq 4\cdot \frac{\miss(\node)}{2}\gamma_{l_{\theta}+2} = 2\miss(\node) \cdot \gamma_{l_{\theta}+2} \approx \miss(\node) \cdot \gamma_{l_{\theta}+1},
    \end{align*}
    as $\gamma_l \approx 2 \gamma_{l-1}$.
    Repeating this process, assuming, in the worst-case, that all subdomains are hot and continue to propagate $\miss(\node)$, it follows that
    \begin{align*}
        \cost_{l_{\theta}+x}(\node) &\leq \max\{2^x,k\}\cdot \frac{\miss(\node)}{2^{x-1}}\gamma_{l_{\theta}+x} \leq 2\miss(\node) \cdot \gamma_{l_{\theta}+x} \approx 2\miss(\node) \cdot \frac{\gamma_{l_{\theta}+1}}{2^{x-1}},
    \end{align*}
    Accumulating these costs, we get
    \begin{align}
        \cost(\node) = \sum_{x<(L-l_{\theta}-1)}\cost_{l_{\theta}+x}(\node) &\lesssim 4 \miss(\node) \cdot \gamma_{l_{\theta}+1}.
        \label{eqn:cost_node}
    \end{align}
    This applies to nodes $\node$ at levels $l_{\theta}\geq \lp$.
    The same analysis extends to nodes $\node$ at levels $l_{\theta} < \lp$, except that the first level of pruning $\miss(\node)$ participates in is $\lp+1$ (not $l_{\theta}+1$).
    Let $\mathcal{H}_l$ denote the hot nodes at level $l\geq \lp$.
    Adding everything together, we arrive at the following bound for the distance between $\treetop$ and $\treeprune$.
    \begin{align*}
         W_1(\treetop, \treeprune))
         &=
         \sum_{l=0}^{\lp-1} \sum_{\theta:\level(\theta)=l} \cost(\node)
         +
         \sum_{l=\lp}^{L-1} \sum_{\node \in \mathcal{H}_l} \cost(\node) \\
         &\lesssim 
         \sum_{l=0}^{\lp-1} \sum_{\theta:\level(\theta)=l} 4 \miss(\node) \cdot \gamma_{\lp+1}
         +
         \sum_{l=\lp}^{L-1} \sum_{\node \in \mathcal{H}_l} 4 \miss(\node) \cdot \gamma_{l+1} \\
         &\leq \sum_{l=0}^{\lp-1} \sum_{\theta:\level(\theta)=l} 4 \miss(\node) \cdot \diam(\subdomain)
         +
         \sum_{l=\lp}^{L-1} \sum_{\node \in \mathcal{H}_l} 4 \miss(\node) \cdot \gamma_{l+1}
    \end{align*}  
    Taking expectations, using Lemma~\ref{lem:miss_error}, with sketch parameter $w=2k$, we arrive at the following. \allowdisplaybreaks
    \begin{align}
        \mathbb{E}[W_1(\treetop,\treeprune)]
        &\lesssim 
        \frac{4}{n}\left(\sum_{l=0}^{\lp-1} \sum_{\theta \in \{0,1\}^l}  \mathbb{E}[\miss(\node)] \cdot \diam(\subdomain)
         +
         \sum_{l=\lp}^{r-1} \sum_{\node \in \mathcal{H}_l}  \mathbb{E}[\miss(\node)] \cdot \gamma_{l}\right) \nonumber  \\ 
        &\leq 
        \frac{4}{n}\left(\sum_{l=0}^{\lp-1} \frac{2\sqrt{2}}{\sigma_{l+1}} \sum_{\theta \in \{0,1\}^l} \diam(\subdomain) + \sum_{l=\lp}^{L-1} \gamma_l \sum_{\node \in \mathcal{H}_l} \frac{2\sqrt{2}j}{\sigma_{l+1}}  + \frac{||\tail_k^{l+1}|| + 2^{-j+1} n}{k}\right) \nonumber \\ 
        &\leq 
        \frac{4}{n}\left(\sum_{l=0}^{\lp-1} \frac{2\sqrt{2} \Gamma_l}{\sigma_{l+1}} + \sum_{l=\lp}^{L-1} \left( \frac{2\sqrt{2} j k}{\sigma_l} + {||\tail_k^{l+1}||_1+ 2^{-j+1} n}\right)  \gamma_l\right) \nonumber \\
        &\leq 
        \frac{4}{n}\left(\sum_{l=1}^{\lp} \frac{2\sqrt{2} \Gamma_{l-1}}{\sigma_l} + \sum_{l=\lp+1}^{L} \left( \frac{2\sqrt{2} j k}{\sigma_l} + {||\tail_k^l||_1} + 2^{-j+1} n \right) \gamma_{l-1}\right), \label{eqn:bound_step_three}
    \end{align}
    where $\Gamma_l = \sum_{\theta \in \{0,1\}^l}\diam(\subdomain)$.
\end{proof}
\noindent
\ref{step:three} bounds the cost of incorrect pruning decisions that occur due to approximations.
Note that the bounds for exact pruning (\ref{step:two}) and approximate pruning (\ref{step:three}) are both expressed in terms of the norm of the tail of the dataset.
This because errors in the frequency approximations from the sketches are also bound by the tail norm of their input vectors.
This demonstrates that sketches (with width $\mathcal{O}(k)$) compose nicely with pruning.

$\treeprune$ has the same structure as $\treehp$, but not the same counts.
\ref{step:four} introduces these (noisy and approximate) counts and, therefore, accounts for the utility loss incurred due to errors in the sampling probabilities.

\begin{restatable}{lemma}{lemstepfour}
    $\mathbb{E} W_1(\treeprune, \treehp) \leq \frac{1}{n}\left( \sum_{l=0}^{\lp} 2\sqrt{2}\frac{\Gamma_{l-1}}{\sigma_l} + \sum_{l=\lp+1}^{L} (2\sqrt{2}\frac{jk}{\sigma_l} + {||\tail_k^l||_1} + \frac{n}{2^{j-1}}) \gamma_{l-1} \right)$
    \label{lem:step4}
\end{restatable}
\noindent
The proof is available in Appendix~\ref{sec:app_step_four} and, similar to Lemma~\ref{lem:step3}, entails quantifying the consistency errors and registering where they occur.
With a bound on the cost of each step in the proof pipeline, we can proceed with the proof of Theorem~\ref{thm:utilityHP}.
\begin{proof}[Proof of Theorem~\ref{thm:utilityHP}]
    By the triangle inequality, the following holds
    \begin{align*}
        \mathbb{E}[W_1(\empiricalx, \treehp)] 
        &\leq
        \mathbb{E}[W_1(\empiricalx, \treetop)] +  \mathbb{E}[W_1(\treetop, \treeprune)] +  \mathbb{E}[W_1(\treeprune, \treehp)]
    \end{align*}
    By Lemmas~\ref{lem:step2}, \ref{lem:step3}, and \ref{lem:step4}, and the prior observation that $||\tail_k^{l-1}||_1\leq ||\tail_k^{l}||_1$, this evaluates as:
    \begin{align*}
        \mathbb{E}[W_1(\empiricalx, \treehp)] 
        &\lesssim 
        \frac{1}{n}\left(\sum_{l=0}^{\lp} \frac{10\sqrt{2} \Gamma_{l-1}}{\sigma_l} + \sum_{l=\lp+1}^{L} \left(\frac{10\sqrt{2} j k}{\sigma_l} + 5 \left({||\tail_k^l||_1} + \frac{n}{2^{j-1}}\right)\right) \gamma_{l-1} + ||\tail_k^L||\sum_{l=\lp+1}^{L-1} \gamma_l \right)  \\
        &\lesssim
        \frac{10\sqrt{2}}{n}\left(\sum_{l=0}^{\lp} \frac{\Gamma_{l-1}}{\sigma_l}  + \sum_{l=\lp+1}^{L} \frac{jk\gamma_{l-1}}{\sigma_l}  \right) + 6\left(\frac{||\tail_k^L||_1}{n} + 2^{-j+1}\right)  \sum_{l=\lp+1}^{L} \gamma_{l-1} 
    \end{align*}
\end{proof}

\section{Proof of Corollary~\ref{cor:hypercube}}
\label{sec:proof_hypercube}

Before proceeding with the proof, we introduce a  useful result
related to the hypercube $\Omega = [0,1]^d$, which is the domain of Corollary~\ref{cor:hypercube}.
The following Lemma bounds the sum of the hypercube subdomain diameters across the pruned levels. 
\begin{lemma}
    On input domain $\Omega=[0,1]^d$, $d \in \mathrm{Z}^{+}$, privacy budget $\varepsilon>0$, and hierarchy depth $L=\log \varepsilon n$ and pruning depth $\lp\geq \log k$,
    \begin{align*}
        {\sum_{l=\lp+1}^{L}}\gamma_{l-1}
        &= \mathcal{O}\left(2^{-\lp/d}\right).
    \end{align*}
    \label{lem:gammal_bound}
\end{lemma}
\begin{proof}
    Let $\Omega=[0,1]^d$ equipped with the $l_{\infty}$ metric.
    The natural hierarchical binary decomposition of $[0,1]^d$ (cut through the middle along a coordinate hyperplane) makes subintervals of length $\diam(\Omega_{\theta}) = \gamma_l \asymp 2^{-l/d}$, for $\theta \in \{0,1\}^l$. 
    Note that $\sum_{l={\lp+1}}^L \gamma_{l-1} = \sum_{l=\lp}^{\log(\varepsilon n)-1} 2^{-l/d}$ is a finite geometric series with common ratio $\rho=2^{-1/d}$.
    Therefore, we can rewrite the sum using the following formula.
    \begin{align}
        \sum_{l=a}^b \rho^l= \frac{\rho^a(1-\rho^{b-a+1})}{1-\rho},
        \label{eqn:geometric}
    \end{align}
    where $b=\log (\varepsilon n)-1$ and $a = \lp \geq \log k$.
    Therefore, we get
    \begin{align*}
        \sum_{l=\lp}^{\log(\varepsilon n)-1} 2^{-l/d} 
        &= 
        2^{-\lp/d}\cdot \frac{1-2^{-(\log \varepsilon n - \lp)/d}}{1-2^{-1/d}}
        \leq 
        k^{-1/d} \cdot \frac{1-\left( \frac{k}{\varepsilon n}\right)^{1/d}}{1-2^{-1/d}}
    \end{align*}
    To evaluate the asymptotics of $\frac{1-\left( \frac{k}{\varepsilon n}\right)^{1/d}}{1-2^{-1/d}}$, we look at its behavior as $d$  and $n$ increase.
    Clearly, this fraction approaches a constant (for fixed $d$) as $n \rightarrow \infty$.
    To determine whether this constant depends on $d$, we need to determine the behavior of the fraction when $d\rightarrow \infty $.
    Both the numerator and the denominator approach $0$ as $d$ increases.
    Therefore, we apply L'Hopital's rule to find the limit.
    Differentiating the numerator, we get
    \begin{align}
        \frac{d}{dd}\left(1-\left(\frac{k}{\varepsilon n}\right)^{1/d} \right)
        &= 
        \frac{\ln\left(\frac{k}{\varepsilon n}\right)\cdot \left(\frac{k}{\varepsilon n}\right)^{1/d}}{d^2}.
        \label{eqn:diff_numerator}
    \end{align}
    Differentiating the denominator, we get
    \begin{align}
        \frac{d}{dd}(1-2^{-1/d}) = -\frac{\ln 2 \cdot 2^{-1/d}}{d^2}
        \label{eqn:diff_denominator}
    \end{align}
    Combining~\eqref{eqn:diff_numerator} and~\eqref{eqn:diff_denominator} with L'Hopital's rule, we get
    \begin{align*}
        \lim_{d \rightarrow \infty} \frac{1-\left( \frac{k}{\varepsilon n}\right)^{1/d}}{1-2^{-1/d}}
        &=
        \lim_{d \rightarrow \infty} - \frac{\ln\left(\frac{k}{\varepsilon n}\right)\cdot \left(\frac{k}{\varepsilon n}\right)^{1/d}}{\ln 2 \cdot 2^{-1/d}} 
        = -\frac{\ln\left(\frac{k}{\varepsilon n}\right)}{\ln 2}
    \end{align*}
    This is a constant for fixed $n$.
    The fraction approaches a constant as either $d$ or $n$ approaches infinity.
    Therefore, for all values of  $d$ and $n$, $\frac{1-\left( \frac{k}{\varepsilon n}\right)^{1/d}}{1-2^{-1/d}}$ is bound above by some value $C\geq 0$. 
\end{proof}
\noindent
We are now ready to proceed with the Proof of Corollary~\ref{cor:hypercube}.
Recall that Corollary~\ref{cor:hypercube} is an extension of Theorem~\ref{thm:utilityHP} on the hypercube.

\corollaryhypercube*
    \begin{proof}
    We set the sketch depth $j = \lceil\log n \rceil$ and the hierarchy depth $L=\log \varepsilon n$.
    Therefore, with sketch width $w=2k$, each sketch occupies $\mathcal{O}(k \log n)$ words of memory.
    As there are at most $L = \mathcal{O}(\log \varepsilon n)$ sketches, the memory occupied by the sketches is $\mathcal{O}(k\log^2 n)$.
    If we set $\lp = \mathcal{O}(\log M) = \mathcal{O}(\log (k\log^2 n)$), then the memory occupied by the tree of exact counts is also $\mathcal{O}(k\log^2 n)$.
    Therefore, the total memory requirement of $\privhp$ is $M = \mathcal{O}(k\log^2 n)$.

    Extending Theorem~\ref{thm:utilityHP} to $\Omega=[0,1]^d$, we proceed by bounding the noise, approximation and resolution terms separately.
    Starting with the approximation term.
    Utilizing Lemma~\ref{lem:gammal_bound} for our choice of $\lp$, 
    For a sufficiently large constant $C_1\geq 0$:
    \begin{align}
        \approximation \leq C_1 \left(\frac{||\tail_k^L||}{n} + 2^{-j}\right) \sum_{l=\lp+1}^L\gamma_{l-1}
        &=
        C_1\left(\frac{||\tail_k^L||}{n} + \frac{1}{n}\right) \mathcal{O}(2^{-\lp/d}) 
        =  \mathcal{O}\left(\frac{||\tail_k^L||}{M^{1/d}n}\right) 
        \label{eqn:bound_approximation}
    \end{align}
    We now look at the noise term separately for $d=1$ and $d\geq 2$.
    For $\Omega=[0,1]$ equipped with the $l_{\infty}$ metric,
    the natural hierarchical binary decomposition of $[0,1]$ makes sub-intervals of length $\diam(\Omega_{\theta}) = \gamma_l = 2^{-l}$ for $\theta \in \{0,1\}^l$.
    Therefore, $\Gamma_l = 1$.
    The following holds for some constant $C_2 \geq 0$.
    \begin{align} 
        \noise^{d=1} \leq  \frac{C_2}{\varepsilon n} \left( \sum_{l=0}^{\lp}  \sqrt{\Gamma_{l-1}}  + \sum_{l=\lp+1}^L \sqrt{kj\gamma_{l-1}} \right)^2 
        &= 
        \frac{C_2}{\varepsilon n} \left( \sum_{l=0}^{\lp} 1 + (k\log n)^{1/2}\sum_{l=\lp+1}^{\log (\varepsilon n)} { 2^{-(\lp-1)/2}} \right)^2 \nonumber\\
        &= \frac{C_2}{\varepsilon n} \left(\lp + (k\log n)^{1/2}\mathcal{O}(2^{-\lp/2}) \right)^2 \nonumber\\
        &= \frac{C_2}{\varepsilon n} \left( \log M + \mathcal{O}\left(\frac{(k\log n)^{1/2}}{(k \log^2 n)^{1/2}}\right) \right)^2 \nonumber\\
        &= \mathcal{O}(\log^2(M)/(\varepsilon n))
        \label{eqn:noise_bound}
    \end{align}
    Combining \eqref{eqn:noise_bound} and \eqref{eqn:bound_approximation}, for $\Omega = [0,1]$, we get:
    \begin{align*}
        \mathbb{E}[W_1(\empiricalx, \treehp)] 
        &= 
        \noise^{d=1} + \approximation 
        =
        \mathcal{O}\left(\frac{\log^2(M)}{\varepsilon n} + \frac{||\tail_k^{\varepsilon n}||}{Mn}\right)
    \end{align*}
    
    For $\Omega=[0,1]^d$, the natural hierarchical binary decomposition of $[0,1]^d$ makes subintervals of length $\diam(\Omega_{\theta}) = \gamma_l \asymp 2^{-l/d}$, for $\theta \in \{0,1\}^l$.
    Therefore, $\Gamma_l =  2^l\cdot 2^{-l/d}= 2^{(1-1/d)l}$.
    We follow the same procedure as above, bounding each term in Theorem~\ref{thm:utilityHP} separately.
    By Lemma~\ref{lem:opt_acc}, \allowdisplaybreaks
    \begin{align}
        \noise^{d\geq 2} \leq
        \frac{C_1}{\varepsilon n} \left( \sum_{l=0}^{\lp}  \sqrt{\Gamma_{l-1}} + \sum_{l=\lp+1}^L \sqrt{kj\gamma_{l-1}} \right)^2
        &= 
        \frac{C_1}{\varepsilon n} \left( \sum_{l=0}^{\lp}  \sqrt{ 2^{(1-1/d)l} } + (k\log n)^{1/2}\sum_{l=\lp+1}^{L} {2^{-l/(2d)}} \right)^2 \nonumber\\
        &=  
        \frac{C_1}{\varepsilon n}\left( \mathcal{O}\left(2^{\frac{1}{2}(1-\frac{1}{d})\lp }\right) + (k\log n)^{1/2}\mathcal{O}(2^{-\lp/(2d)}) \right)^2 \nonumber\\
        &= 
        \mathcal{O}\left(\frac{M^{(1-\frac{1}{d})/2}}{\varepsilon n} + \frac{(k\log n)^{1/2}}{\varepsilon n (k \log^2 n)^{1/(2d)}} \right)^2 \nonumber\\
        &=\mathcal{O}\left(\frac{M^{(1-\frac{1}{d})}}{\varepsilon n}\right)
        \label{eqn:01d_noise}
    \end{align}
    The second line comes from Lemma~\ref{lem:gammal_bound}, by setting the dimension to $2d$.
    Combining \eqref{eqn:01d_noise} and \eqref{eqn:bound_approximation}, for $\Omega = [0,1]^d$, we get
    \begin{align*}
        \mathbb{E}[W_1(\empiricalx, \treehp)]
        = \noise^{d\geq 2} + \approximation 
        &=
        \mathcal{O}\left( \frac{M^{(1-\frac{1}{d})}}{\varepsilon n}+  \frac{||\tail_k^{\varepsilon n}||}{M^{1/d}n}\right)
    \end{align*}
    
    For the time complexity, at each update, $x\in \mathcal{X}$ performs a root to leaf traversal, updating the counter for each node.
    Both exact counters are updated in constant time and approximate counters are updated in $\mathcal{O}(\log n)$.
    Therefore the cost of an update is $\mathcal{O}(L \cdot \log n)= \mathcal{O}(\log^2 n)$.

    The partition tree is built one level at a time.
    The noisy frequency estimates for each node can each be retrieved in $\mathcal{O}(j) = \mathcal{O}(\log n)$ time.
    At the first level of pruning $\lp$, $\mathcal{O}(2^{\lp}) = \mathcal{O}(M)$ candidates are retrieved in $\mathcal{O}(M \log n)$ time.
    Then $\mathcal{O}(M)$ consistency steps are performed, each in constant time.
    To retrieve the top-$k$, they are sorted in $\mathcal{O}(M \log M) = o (M \log n)$ time.
    
    The remaining levels $\mathcal{O}(L-\lp)= \mathcal{O} (\log n) $ levels output $2k$ candidates.
    Each frequency estimate is retrieved in $\mathcal{O}(\log n)$ and made consistent in constant time.
    Once the estimates are retrieved, they need to be sorted so that the bottom-$k$ can be pruned.
    Sorting takes $\mathcal{O}(k \log k)$ time.
    Therefore each remaining level is processed in $\mathcal{O}(k \log n)$ time.
    As there are $\mathcal{O}(\log n)$ levels $l>\lp$, this process requires $\mathcal{O}(k \log^2 n )$ time.
    Therefore, the whole process completes in $\mathcal{O}(M \log n)$ time.

\end{proof}

\section{Conclusion}
\label{sec:conclusion}

In summary, $\privhp$ offers a novel approach to differentially private synthetic data generation by leveraging a hierarchical decomposition of the input domain within bounded memory constraints. 
Unlike existing methods, $\privhp$ provides a principled trade-off between accuracy and space efficiency, balancing hierarchy depth, noise addition, and selective pruning to preserve high-frequency subdomains.
Our theoretical analysis establishes rigorous utility bounds, demonstrating that $\privhp$ achieves competitive accuracy with significantly reduced memory usage compared to state-of-the-art methods.
By introducing the pruning parameter $k$, our approach enables fine-grained control over the trade-off between space and utility, making $\privhp$ a flexible and scalable solution for private data summarization in resource-constrained environments.

\bibliographystyle{unsrt}  
\bibliography{bib}

\appendix

\section{Proof of Lemma~\ref{lem:cms_expected_error}}
\label{sec:app_sketch}

We assume that the hash functions are fully random.
\lemsketch*
\begin{proof}
    Fix an index $x$ and let $\textsf{Top-}w(v)$ denote the set of coordinates in $v$ with the $w$ highest magnitudes.
    For a given level $ i \in [j]$ we can bound the probability that $x$ does not collide with an index in $\textsf{Top-}w(v)$:
    \begin{align*}
        \textsf{Pr}[ h_i(x) \neq h_i(y), \forall y \in \textsf{Top-}w(v)\setminus \{x\} ] 
        &= \prod_{y \in \textsf{Top-}w(v)\setminus \{x\}} \textsf{Pr}[h_i(x) \neq h_i(y)] \\
        &= \left( \frac{2w-1}{2w} \right)^w, 
    \end{align*} 
    as the sketch has width $2w$. 
    Taking the logarithm of both sides we get, 
    \begin{align*}
        \ln \textsf{Pr}[ h_i(x) \neq h_i(y), \forall y \in \textsf{Top-}w(v)\setminus \{x\} ] 
        &= w \ln \left(1- \frac{1}{2w} \right) \\
        &> w \cdot \left(- \frac{1}{2w} - \left(\frac{1}{2w}\right)^2\right) 
        = -\frac{1}{2} - \frac{1}{4w},
    \end{align*}
    where the inequality comes from the fact that $\ln(1-z) > -z - z^2 $ for $z\in(0,0.5)$.
    For $w\geq 1$, we have $-\frac{1}{2} - \frac{1}{4w} > - \ln 2$.
    Thus,
    \begin{align*}
      \textsf{Pr}[ h_i(x) \neq h_i(y), \forall y \in \textsf{Top-}w(v)\setminus \{x\} ] &>e^{-\ln(2)} = \frac{1}{2}.
    \end{align*}
    Therefore, the probability that an index does not collide with a $\textsf{Top-}w$ index is greater than a half. 
    Using a standard argument with Chernoff bounds, across $ j $ levels, in \textit{at least one level} $x$ has no collisions with the $\textsf{Top-}w$ indices with probability greater than $1-2^{-j}$.

    We now look at the probability that some some $y \in \textsf{Top-}w(v)\setminus \{x\}$ observes a collision with $x$, conditioned on the event that at least one index in $ \textsf{Top-}w(v)\setminus \{x\}$ collides with $x$.
    Let $\mathcal{E}_{i}$ denote the event: $\exists y^{\prime}\in \textsf{Top-}w(v)\setminus \{x\} : h_{i}(x)=h_{i}(y^{\prime})$.
    Then $\forall y \in \textsf{Top-}w(v)\setminus \{x\}$,
    \begin{align}
        \textsf{Pr}[h_i(x) = h_i(y) | \mathcal{E}_i] 
        & = \frac{\textsf{Pr}[h(i(x) = h_i(y) \cap \mathcal{E}_i] }{\textsf{Pr}[ \mathcal{E}_i] } 
        = \frac{1/(2w)}{(1-1/(2w))^w} \approx  \frac{1/(2w)}{e^{-1/2}} < \frac{2}{w}
        \label{eqn:coll_bound}
    \end{align}
    Let $i^*$ represent the level where $x$ has the fewest collisions with the $\textsf{Top-}w$ indices and $\hat{v}_x^*$ denote the estimate taken from level $i^{*}$.
    As $i^*$ has the fewest collisions, it holds that, $\forall y \in \textsf{Top-}w(v)\setminus \{x\}$ and $ i \neq i^{*}$, $\textsf{Pr}[h_{i^*}(x) = h_{i^*}(y) | \mathcal{E}_{i^*}] \leq \textsf{Pr}[h_{i}(x) = h_{i}(y) | \mathcal{E}_{i}]$ and can be bound by~\eqref{eqn:coll_bound}.  
    It follows that,
    \begin{align*}
        \mathbb{E}[\hat{v}_x] 
        &\leq \mathbb{E}[\hat{v}^*_x] \\
        &\leq v_x + \mathbb{E}\left[\sum_{y \in \textsf{tail}_w(v)} \mathbbm{1}[h_{i^*}(x) = h_{i^*}(y)]\cdot v_y\right] + \textsf{Pr}[\mathcal{E}_{i^*}] \left(\mathbb{E}\left[\sum_{y \in \textsf{head}_w(v)} \mathbbm{1}[h_{i^*}(x) = h_{i^*}(y) \mid ]\cdot v_y \mid \mathcal{E}_{i^*} \right] \right) \\
        &\leq v_x + \sum_{y \in \textsf{tail}_w(v)} \textsf{Pr}[h_{i^*}(x) = h_{i^*}(y)] \cdot v_y + 2^{-j} \cdot \sum_{y \in \textsf{head}_w(v)} \textsf{Pr}[h_{i^*}(x) = h_{i^*}(y) \mid \mathcal{E}_{i^*}] \cdot v_y \\
        &= v_x + \frac{||\tail_w(v)||_1}{2w} + 2^{-j} \cdot \frac{2\cdot \sum_{y \in \textsf{head}_w(v) } v_y}{w} \\
        &\leq v_x + \frac{\sum_{y \in \textsf{tail}_w(v) } v_y + 2^{-j+1} ||v||_1}{w}.
    \end{align*} 
\end{proof}
Subtracting $v_x$ from both sides completes the proof.

\section{Proof of Lemma~\ref{lem:opt_acc}}
\label{sec:app_opt_sigma}

We now evaluate the optimal allocation of the privacy budget across levels in the hierarchy.
\thmoptacc*

\begin{proof}
    Following \cite{he2023algorithmically}, we will use Lagrange multipliers to find the optimal choices of the $\{\sigma_l\}$.
    With a partition of depth $L$, we are subject to a privacy budget of $\varepsilon = \sum_{l=0}^{L} \sigma_l$.
    Therefore, as we aim to minimize the accuracy bound subject to this constraint, we end up with the following optimization problem.
    \begin{align*}
        \min \mathbb{E}[W_1(\empiricalx, \treehp)] \qquad \text{s.t} \; \varepsilon = \sum_{l=0}^{L} \sigma_l.
    \end{align*}
    With parameters $n,\varepsilon, k,L$ fixed in advance and $\delta_L$ dependent only on $L$, this optimization problem is equivalent to 
    \begin{align*}
        \min \left(\sum_{l=0}^{\lp}  \frac{\Gamma_{l-1}}{\sigma_l} + \sum_{l=\lp+1}^{L}\frac{jk\gamma_{l-1}}{\sigma_{l}}\right) \qquad \text{s.t} \; \varepsilon = \sum_{l=0}^{L} \sigma_l.
    \end{align*}
    Now, consider the Lagrangian function
    \begin{align*}
        f(\sigma_0,\ldots, \sigma_l, t)&:= \left( \sum_{l=0}^{\lp}  \frac{\Gamma_{l-1}}{\sigma_l}  + \sum_{l=\lp+1}^{L}  \frac{jk\gamma_{l-1}}{\sigma_{l}}\right) - t\left(  \sum_{l=0}^{L} {\sigma_l} -\varepsilon\right),
    \end{align*}
    and the corresponding equation
    \begin{align*}
       \frac{\delta f}{\delta \sigma_0} = \cdots =  \frac{\delta f}{\delta \sigma_r} =  \frac{\delta f}{\delta t} = 0.  
    \end{align*}
    One can easily check that the equations have the following unique solution
    \begin{align}
        \sigma_l 
        &=
        \begin{cases}
            \frac{\varepsilon \sqrt{\Gamma_{l-1}}}{S} & l\leq \lp \\
            \frac{\varepsilon \sqrt{jk\gamma_{l-1}}}{S} &  \text{ Otherwise}
        \end{cases}
        \qquad \text{where} \; S= \sum_{l=0}^{\lp}  \sqrt{\Gamma_{l-1}} + \sum_{l=\lp+1}^{L} \sqrt{jk\gamma_{l-1}}.
        \label{eqn:opt_sigma}
    \end{align}
    This states that the amount of noise per level is inversely proportional to its effect on the utility of the partition.
    Substituting the optimized values of $\{\sigma_l\}$ into $\noise$ in~\eqref{eqn:utility_bound}, we get, for some constant $C\geq 0$,
    \begin{align*}
        \noise  
        &\leq 
        \frac{C}{n}\left(\sum_{l=0}^{\lp} \frac{\Gamma_{l-1}}{\sigma_l} + \sum_{l=\lp+1}^{L} \frac{jk\gamma_{l-1}}{\sigma_{l}} \right)  \\
        &=
        \frac{C}{n} \left(\sum_{l=0}^{\lp}\frac{S \Gamma_l}{\varepsilon \sqrt{\Gamma_{l-1}}} + \sum_{l=\lp+1}^{L} \frac{Sjk\gamma_l}{\varepsilon \sqrt{jk\gamma_{l-1}}} \right)  \\ 
        &= \frac{C\cdot S}{\varepsilon n} \left( \sum_{l=0}^{\lp}  \sqrt{\Gamma_{l-1}} + \sum_{l=\lp+1}^{L} \sqrt{jk\gamma_{l-1}} \right)  \\
        &= \frac{C}{\varepsilon n} \left( \sum_{l=0}^{\lp}  \sqrt{\Gamma_{l-1}} + \sum_{l=\lp+1}^{L} \sqrt{jk\gamma_{l-1}} \right)^2,
    \end{align*}
    which completes the proof.
\end{proof}

\section{Proof of Lemma~\ref{lem:miss_error}}
\label{sec:app_proofs6}
\lemmiss*
\begin{proof}
    We will continue our accounting approach, disaggregating the consistent counts in child nodes into exact counts, the consistency error and errors higher in the hierarchy.
    Depending on whether error correction is used during consistency, we have three cases to consider:
    \begin{enumerate}[label=\texttt{Case (\arabic*)}, leftmargin=2cm]
        \item\label{case:one} No error correction is used;
        \item\label{case:two} $\mathtt{Correction 1}$ is used (Algorithm~\ref{alg:consistency} Line~\ref{line:correction1});
        \item\label{case:three} $\mathtt{Correction 2}$ is used (Algorithm~\ref{alg:consistency} Line~\ref{line:correction21}).
    \end{enumerate}
    $\mathtt{Correction 1}$ and $\mathtt{Correction 2}$ have the effect of reducing the amount of error in the node counts.
    Therefore, they cannot increase the number of misses in a node.
    We prove this notion formally and consider each case separately.
    \paragraph*{\ref{case:one}}
    When no error correction is used, $\miss(\node)$ is defined in~\eqref{eqn:miss_definition}.
    Taking expectation, we get
    \begin{align}
        \mathbb{E}[ \miss(v_{\theta})] &= \mathbb{E}[|(\lambda_{\theta0}- \lambda_{\theta1} + e_{\theta0}-e_{\theta1})/2|] \nonumber \\
        &\leq \frac{1}{2}(\mathbb{E}[|\lambda_{\theta0}- \lambda_{\theta1}|] + \mathbb{E}[|e_{\theta0}-e_{\theta1}|]) \label{eqn:bound_conserr} \\
        &\leq \mathbb{E}[ \max \{|\lambda_{\theta0}|, |\lambda_{\theta1}|\}] + \frac{1}{2} (\mathbb{E}[|e_{\theta 0}|]  + \mathbb{E}[|e_{\theta 1}]) \nonumber \\
        &\leq \mathbb{E}[ \max \{|\lambda_{\theta0}|, |\lambda_{\theta1}|\}] + \frac{||\textsf{tail}_w^{l+1}||}{2w}  + 2^{-j+1}n \label{eqn:miss_ubound}   
    \end{align}
    The third inequality follows from Lemma~\ref{lem:cms_expected_error}.
    As the $\lambda_{\theta0}$ and $\lambda_{\theta 1}$ are independent Laplace variables with noise defined in~\eqref{eqn:laplace_parameters}, the following inequality completes the upper bound for case 1.
    \begin{align*}
        \mathbb{E}[\max \{|\lambda_{\theta0}|, |\lambda_{\theta1}|\}] \leq 
        \begin{cases}
             2\sqrt{2}\sigma_l^{-1} & l\leq \lp \\
              2\sqrt{2}\sigma_l^{-1} \cdot j & \text{Otherwise}
        \end{cases}
    \end{align*}
    \paragraph*{\ref{case:two}}
    We focus on a correction made to $v_{\theta0}$.
    A parallel argument can be made for $v_{\theta1}$.
    In $\mathtt{Correction 1}$, $v_{\theta0}.\counter^{\text{before}}$ is set to $0$ if it is negative.
    As $c_{\theta0}, e_{\theta0}\geq 0$, this can only happen if $\lambda_{\theta 0}<0$.
    Therefore, under our accounting approach, the correction $v_{\theta0}.\counter^{\text{before}} \gets 0$ is made possible if $\lambda_{\theta 0}$ is changed to some value $|\lambda_{\theta 0}^{\prime}|\leq |\lambda_{\theta 0}|$.
    Inserting this value into Inequality~\eqref{eqn:miss_ubound} has the effect of reducing the bound on the number of misses.
    \paragraph*{\ref{case:three}}
    As above, we focus on a correction made to $v_{\theta0}$.
    A parallel argument can be made for $v_{\theta1}$.
    $\mathtt{Correction 2}$ is triggered on node $v_{\theta 0}$, when 
    \[
    v_{\theta 0}.\counter^{\text{before}} - \Lambda/2 < 0.
    \]
    By~\eqref{eqn:misses_consistency}, this implies that
    \begin{align}
        c_{\theta0} + (\lambda_{\theta0}- \lambda_{\theta1})/2 + (e_{\theta0}-e_{\theta1})/2 +\err_{\theta}/2 < 0.
        \label{eqn:correction2}
    \end{align}
    Therefore, consistency is violated when any combination of $(\lambda_{\theta0}- \lambda_{\theta1})$, $(e_{\theta0}-e_{\theta1})$ and $\err_{\theta}$ are non-positive and sufficiently large.
    The error correction step entails setting
    \[
    v_{\theta 0}.\counter^{\text{after}} \gets 0, 
    \]
    thus, \textit{reducing} the amount of error in $v_{\theta 0}.\counter^{\text{after}}$.
    Following our accounting approach, this can be achieved through rescaling $\Lambda/2$ by introducing new error terms $\lambda_{\theta 0}^{\prime}$ or $e_{\theta0}^{\prime}$, with
    \begin{align*}
        |\lambda_{\theta0}^{\prime}- \lambda_{\theta1}|
        &\leq 
        |\lambda_{\theta0}- \lambda_{\theta1}| \\
        |e_{\theta0}^{\prime}-e_{\theta1}|
        &\leq
        |e_{\theta0}-e_{\theta1}|,
    \end{align*}
    such that~\eqref{eqn:correction2} no longer holds.
    By inserting these values into ~\eqref{eqn:bound_conserr}, we reduce the bound on the consistency error.
    Therefore,~\eqref{eqn:miss_ubound} holds in all three cases.
\end{proof}

\section{Proof of Lemma~\ref{lem:step4}}
\label{sec:app_step_four}
The proof relies on the following result.
\begin{lemma}[\cite{he2023algorithmically}]
    For any finite multisets $U \subseteq V$ such that all elements in $V$ are from $\Omega$, one has
    \[
    W_1(\mu_{ U}, \mu_{V}) \leq \frac{|V \setminus U|}{|V|} \cdot {\diam}(\Omega).
    \]
    \label{lem:setminus_wass}
\end{lemma}
\noindent
A main component of the proof entails quantifying the consistency errors and registering where they occur.
\begin{proof}[Proof of Lemma~\ref{lem:step4}]
    This proof is based on the proof of Theorem~$10$ in \cite{he2023algorithmically}.
    For root node $v_{\varnothing}\in \treehp$, let $m=v_{\varnothing}.\counter$ denote the number of ``points'' in $\treehp$ (this number might be a decimal), where a point refers to a unit of probability mass.
    Moving from $\treeprune$ to $\treehp$ can be done in two steps:
\begin{enumerate}
    \item Transform the $n$ point tree $\treeprune$ to the $m$ point tree $\treeprune^{\prime}$ by adding or removing $|n-m|$ points\footnote{This step introduces an additional miss. 
    We did not include this miss in the approximate pruning step (Lemma~\ref{lem:step3}) as it is evenly distributed among all descendants.}.
    \item Transform $\treeprune^{\prime}$ to $\treehp$ by recursively moving $\miss(v_{\theta})$ points between sibling nodes $v_{\theta 0}$ and $v_{\theta 1}$ and propagating each point down to a leaf node.
\end{enumerate}
With step 2, the total distance points move is at most
\begin{align}
    \sum_{l=0}^{r-1} \sum_{\node \in \mathcal{H}_l} \miss(v_{\theta})\cdot \diam(\subdomain) 
    &:= 
    C
    \label{eqn:C_def}
\end{align}
Therefore, since $|\treehp| =m$, it follows that
\begin{align}
    W_1(\treeprune^{\prime}, \treehp) &\leq \frac{C}{m}.
    \label{eqn:consistency_step}
\end{align}
Recall that the first step transforms the tree $\treeprune$ of size $n$ to the tree $\treeprune^{\prime}$ of size $m = n+\mathtt{sign}(\lambda_{\varnothing}) |\lambda_{\varnothing}|$, by adding or removing points.
For $\lambda_{\varnothing} \geq 0$, $\treeprune^{\prime}$ is created by \textit{adding} $\lambda_{\varnothing} $ points. 
Therefore, by Lemma~\ref{lem:setminus_wass}, it follows that
\begin{align*}
    W_1(\treeprune,\treeprune^{\prime}) &\leq \frac{\lambda_{\varnothing}}{m} \cdot \diam(\Omega).
\end{align*}
Combining this with (\ref{eqn:consistency_step}), we get
\begin{align*}
    W_1(\treeprune, \treehp)
    &\leq 
    \frac{\lambda_{\varnothing}\Gamma_{-1} + C}{m} 
    \leq 
    \frac{\lambda_{\varnothing}\Gamma_{-1} + C}{n},
\end{align*}
where $\Gamma_{-1}=\Gamma_0$.
Alternatively, for $\lambda_{\varnothing}>0$, $\treeprune^{\prime}$ is obtained from $\treeprune$ by removing a set $\mathcal{X}_0$ of $|n-m|$ (possibly fractional) points from $\treeprune$.
As previously stated, $\treehp$ is constructed from $\treeprune^{\prime}$ by moving points distance $C$.
Therefore, $\treehp\cup \mathcal{X}_0$ can also be constructed from $\treeprune^{\prime}$ by moving points distance $C$ (the $\mathcal{X}_0$ points remain unmoved).
Since $|\treeprune^{\prime}|=n$, it follows that
\begin{align*}
    W_1(\treeprune, \treehp\cup \mathcal{X}_0) &\leq \frac{C}{n}.
\end{align*}
Further, Lemma~\ref{lem:setminus_wass} gives:
\begin{align*}
    W_1(\treehp, \treehp\cup \mathcal{X}_0) &\leq \frac{|\mathcal{X}_0|}{|\treehp\cup \mathcal{X}_0|}\cdot \diam(\Omega) \leq \frac{|\lambda_{\varnothing}|\Gamma_{-1}}{n}.
\end{align*}
Combining the two bounds by the triangle inequality, we get
\begin{align*}
    W_1(\treeprune,\treehp) &\leq \frac{|\lambda_{\varnothing}|\Gamma_{-1} + C}{n}
\end{align*}
In other words, the bound holds in both cases. 
Recalling the definition of $C$ from \eqref{eqn:C_def}, 
\begin{align*}
    \mathbb{E}[W_1(\treeprune,\treehp)] 
    &\leq 
    \frac{1}{n}\left(\mathbb{E}[\lambda_{\varnothing}]\Gamma_{-1} + \sum_{l=0}^{L-1} \sum_{\node \in \mathcal{H}_l} \mathbb{E}[\miss(\node)]\cdot \diam(\subdomain) \right) \\
    &\leq 
    \frac{1}{n}\left(2\sqrt{2}\sigma_0^{-1}\Gamma_{-1} + \sum_{l=1}^{\lp} 2\sqrt{2}\sigma_l^{-1} \Gamma_{l-1} + \sum_{l=\lp+1}^{L} \left(\frac{2\sqrt{2}jk}{\sigma_l} + {||\tail_k^{l+1}||_1+ 2^{-j+1} n} \right) \gamma_{l-1}  \right) \\
    &=
    \frac{1}{n}\left( \sum_{l=0}^{\lp} \frac{2\sqrt{2} \Gamma_{l-1}}{\sigma_l} + \sum_{l=\lp+1}^{L} \left(\frac{2\sqrt{2}jk}{\sigma_l} + {||\tail_k^{l+1}||_1+ 2^{-j+1} n}  \right) \gamma_{l-1} \right),
\end{align*}
where the second inequality follows from~\eqref{eqn:bound_step_three}.
\end{proof}

\end{document}